%% file: ms.tex
\documentclass[11pt,final]{article}
\usepackage{fullpage}
\usepackage[numbers, compress]{natbib}

\usepackage{lmodern}
\usepackage[utf8]{inputenc} 
\usepackage[T1]{fontenc}    
\usepackage{hyperref}       
\usepackage{url}            
\usepackage{booktabs}       
\usepackage{amsfonts}       
\usepackage{nicefrac}       
\usepackage{microtype}      
\usepackage{enumitem}

\usepackage[edges]{forest}
\usepackage{algorithm}
\usepackage[noend]{algorithmic}
\usepackage{graphicx}
\usepackage[utf8]{inputenc} 
\usepackage{url}            
\usepackage{hyperref}       
\usepackage{nicefrac}       
\usepackage{microtype}      
\hypersetup{
     colorlinks  = true,
     urlcolor    = teal,
	 citecolor   = teal,
	 linkcolor   = red
}

\usepackage{enumitem}
\usepackage{multicol}

\input{defs}

\newcommand{\nonl}{\renewcommand{\nl}{\let\nl\oldnl}}
\newcommand{\ise}{\textsc{Ibs-P}}

\newcommand{\ALG}{\small\textnormal{\textsf{ALG}}}
\newcommand{\OPT}{\small\textnormal{\textsf{OPT}}}
\newcommand{\OGP}{\small\textnormal{\textsf{OGP}}}
\newcommand{\SEGB}{\small\textnormal{\textsf{SEGB}}}
\title{Learning under Invariable Bayesian Safety}

\author{
\and Gal Bahar%
\thanks{%
    {Technion - Israeli Institute of Technology (\url{bahar@campus.technion.ac.il})}}
\and Omer Ben{-}Porat%
\thanks{%
    {Technion - Israeli Institute of Technology (\url{omerbp@campus.technion.ac.il}). Contact author.}}
\and Kevin Leyton{-}Brown%
\thanks{%
    {University of British Columbia, Canada (\url{kevinlb@cs.ubc.ca})}}
\and Moshe Tennenholtz%
\thanks{%
    {Technion - Israeli Institute of Technology (\url{moshet@ie.technion.ac.il})}}
}

\begin{document}

\maketitle

\begin{abstract}
A recent body of work addresses safety constraints in explore-and-exploit systems. Such constraints arise where, for example, exploration is carried out by individuals whose welfare should be balanced with overall welfare. In this paper, we adopt a model inspired by recent work on a bandit-like setting for recommendations. We contribute to this line of literature by introducing a safety constraint that should be respected in every round and determines that the expected value in each round is above a given threshold. Due to our modeling, the safe explore-and-exploit policy deserves careful planning, or otherwise, it will lead to sub-optimal welfare. We devise an asymptotically optimal algorithm for the setting and analyze its instance-dependent convergence rate.
\end{abstract}

\section{Introduction}\label{sec:intro}
Explore-and-exploit tradeoffs are central to Machine Learning (ML) and Artificial Intelligence (AI). They are the core problems in well-studied decision-making problems like Markov Decision Processes \cite{bellman1957markovian} and Multi-armed Bandits (MABs) \cite{bubeck2012regret}. Such dilemmas refer to environments in which we can \textit{exploit} the knowledge already acquired and get the outcome that we expect, or \textit{explore} other alternatives and reveal potentially better outcomes. Motivated by the extensive use of algorithmic decision making, exploration under system constraints is becoming increasingly popular. Some of these constraints arise due to fairness considerations \cite{MatthewKearnsMorgensternRothNIPS2016,liu2017calibrated},  the need for privacy \cite{tossou2016algorithms,tossou2017achieving}, mitigating  strategic behavior \cite{Bahar2016,Mansour2015,Mansour2016Slivkins,MansourSW18} and safety (see \citet{garcia2015comprehensive} for a recent survey and taxonomy). In this work, we contribute to the literature on learning under safety concerns.

Since safety in this context is ambiguous, it is no surprise that many definitions and approaches exist. In one line of work, the optimization criterion is transformed to address other factors such as variance \cite{howard1972risk}, the worst outcome \cite{borkar2002q, heger1994consideration}, or the probability of visiting unwanted states \cite{geibel2005risk}. In another line of work, the exploration process itself is constrained to address safety considerations \cite{kazerouni2017conservative,wu2016conservative}.  Further, recent work \cite{amani2019linear,usmanova2019safe} suggests that safety should be applied in \textit{almost} every action algorithms make, despite that the safe set of actions is initially unknown. Of course, since the safe set of actions is learned throughout the process,  safety constraints can only be guaranteed to be satisfied with high probability.


In this paper, we contribute to the latter strand of research. We consider a MAB model under several assumptions, and require what is perhaps the strictest safety constraint that nevertheless allows learning: that the expected reward in every round exceeds some safety threshold, given the Bayesian information available to the algorithm. This safety constraint complements prior work in two important ways. First, safety should hold \textit{invariably}, i.e., in every single round. This is crucial in, e.g., scenarios where a MAB algorithm experiments with one agent per round. Second, the safety constraint relies on the \textit{full Bayesian information} the algorithm has, in sharp contrast to work that leverages information asymmetry (see related work below). 
Our safety constraint can be used to strengthen existing safety requirements like Bayesian Incentive Compatibility \cite{Bahar2016,Mansour2015,Mansour2016Slivkins,MansourSW18}, or to relax legal responsibilities that allow no learning whatsoever (like fiduciary duty). Our investigation broadens the safety spectrum in MABs, and demonstrates that while learning in the invariable Bayesian safety setting is non-trivial, it can still occur.

\subsection{Our Contribution}
We consider a MAB setting with a Bayesian prior and static rewards; that is, the reward of each arm is initially unknown, but is realized only once.\footnote{While this assumption is limiting, it allows us to focus on the novel viewpoint of invariable safety we propose in this paper. Indeed, a similar approach is taken by much other recent work proposing incentive compatibility constraints in MABs \cite{Bahar2016,cohen2019optimal,Mansour2015}. As we will see, even given this assumption, the problem is technically challenging.} Some of the arms may be risky, i.e., hide negative rewards, but can also be highly rewarding. Additionally, we assume that there is a \textit{safe arm}, with a known deterministic reward. In every round, the decision maker (henceforth DM) picks a distribution over the arms, which we term a \textit{portfolio}. After committing to a portfolio, DM samples an arm according to the weights of the portfolio and plays it. The utility of DM is the sum of rewards, which she wishes to maximize. Without any further constraints, this is straightforward: try every arm once, and then repeatedly choose the best arm. (Due to our static rewards assumption, one round of exploration is enough to reveal all the rewards.) However, we limit DM to playing only \textit{safe portfolios}; a portfolio is  safe if its expected reward is weakly better than that of the safe arm, 
where this expectation is taken over the Bayesian information and the weights it gives to arms.

We contribute to the existing literature both conceptually and technically. Conceptually, safety is applied in \textit{every} round, and not only with high probability like previous work (see, e.g., \citet{amani2019linear,usmanova2019safe}). Our approach to safety is highly desired in, e.g., situations where DM interacts with agents. We ask that DM provides \textit{every} agent a portfolio that is weakly better than the safe arm in expectation, independently of the other agents. This is in sharp contrast to traditional MABs, in which algorithms intentionally devote some rounds to exploration. Our safety constraint spreads exploration over rounds to distribute its cost, and provides a strong guarantee for agents to use the system (as it is weakly better than the safe arm). Our main conceptual take-away is that systems and decision-makers that are required to ensure safety in its (arguably) most stringent form can still enjoy learning.

On the technical side, the main algorithmic question is to maximize the expected utility subject to safety. Since rewards are realized only once (and remain fixed throughout the execution), planning is crucial: arms with high a priori value extend the set of safe portfolios and enable exploration of a priori inferior arms. 
Our main technical contribution is an asymptotically optimal algorithm, under a stochastic dominance assumption on the rewards. Our approach is based on a constrained Goal Markov Decision Process (GMDP) \cite{barto1995learning}. We devise a neat optimal policy for this GMDP, and harness it to craft an asymptotically optimal algorithm. Finally, we analyze the instance-dependent convergence rate.

\subsection{Related Work}\label{subsec:related}
As argued above, the notion of safety has many diverse definitions in the ML and AI literature. For example, \citet{Moldovan:2012} introduce an algorithm that
allows safe exploration in Markov Decision Processes in order to avoid fatal absorbing states. Our work takes a different tack, as it constrains the exploration process rather than modifying the objective criterion.

Several lines of work have applied safety constraints to Multi-Arm Bandit problems. One considers a MAB problem with a global budget where the objective is to maximize the total reward before the learner runs out of resources (see, e.g., \cite{agrawal2016linear,badanidiyuru2013bandits,badanidiyuru2014resourceful,wu2015algorithms}). Another considers stage-wise safety,  ensuring that regret performance stays above a threshold determined by a baseline strategy at every round \cite{kazerouni2017conservative,wu2016conservative}. Notably, in these lines of work, the constraint applies to the cumulative resource consumption/reward across the entire run of the algorithm. In contrast, \citet{amani2019linear} consider the application of a reward constraint at each round. In their work, the set of safe decisions is left uncertain, given the inherent uncertainty in the learning system. They therefore have a two-fold objective: to minimize regret while learning the safe decision set.

Our work is motivated by the idea of using a bandit-like setting for recommendations \cite{Fiduciary,BaharST19,cohen2019optimal,KremerMP13}. Under this perception, each round serves a different agent and, therefore, the safety constraint should be applied in each round. This strengthens the approach of \citet{amani2019linear}, as we require the constraint to {\em always} hold (and not only with high probability). As in \citet{Fiduciary,BaharST19,cohen2019optimal,KremerMP13}, we assume each arm is associated with a fixed value, which is initially unknown, selected from a known distribution. Given this, our problem becomes the careful planning of safe exploration rather than the more standard explore and exploit subject to constraints, with probabilistic guarantees. 
Our setting conceptually departs from incentive compatible MABs \cite{BaharST19,cohen2019optimal,KremerMP13}. In that line of research, optimal solutions rely on \textit{information asymmetry}: each agent has only a priori knowledge while the algorithm witnesses the rewards of all preceding agents; thus, the algorithm can induce agents to explore a priori inferior arms. While this approach might be suitable for some scenarios, it can be highly undesired in cases where, e.g., the system can be held liable to its actions, or cares for long-term engagement. Indeed, this manipulation is what we try
to remedy in the present paper.

Another interesting view of safety is to assume the underlying system is safety-critical and to present active
learning frameworks that use Gaussian Processes as non-parametric models to learn the safe decision set (e.g., \cite{sui2015safe,sui2018stagewise}). So, the safe decision set is fixed but is initially unknown, and is learned in the process. Of course, we never know the safe decisions for certain in such settings. This is unlike our work, where we do know the safety constraint at each point. Safety depends on our knowledge, which is perfect in the Bayesian sense. 

\section{Problem Statement}\label{sec:problem statement}

In this section, we formally define the Invariable Bayesian Safety problem ($\ise$ for shorthand). We consider a set $A$ of $K$ arms, $A=\{a_1,\dots a_K\}$. The reward of arm $a_i$ is a random variable $X(a_i)$, and $(X(a_i))_{i=1}^K$ are mutually independent. The rewards are static, i.e., they are realized only once, but are initially unknown. We denote by $\mu({a_i})$ the expected value of $X(a_i)$, namely $\mu(a_i)\defeq\E\left[X(a_i)\right]$. DM knows the Bayesian priors, i.e., distribution $X(a_i)$ for every $a_i\in A$. There are $T$ rounds, where we address rounds as interactions with agents; namely, DM makes a decision for agent $t$ in round $t$. We augment the set of arms with a \textit{safe arm}, whose reward is always 0,\footnote{The selection of zero as the threshold is arbitrary; we can work with any real scalar similarly.} and which we denote by $a_0$ for simplicity. In each round, DM can pick the safe arm, or an action from $A$. We also let $A^+ \defeq A\cup\{a_0\}$, and let $\Delta(A^+)$ be the set of distributions over the elements of $A^+$.

We denote by $\mI_t$ the information DM has at time $t$---the prior distributions and the reward realizations acquired until round $t$. Namely, $\mI_1$ encodes the prior information solely, $\mI_2$ encodes both the prior information and the reward of the arm selected at round 1, and so on. If the reward $X(a_i)$ was realized before round $t$ and is hence known to DM, we use $x(a_i)$ to denote its value.

In every round, DM plays a \textit{portfolio} of arms. A portfolio is an element from $\Delta(A^+)$---it is a distribution over the arms and the safe arm. Put differently, a portfolio is the extension of actions to  randomized actions. Whenever DM picks a portfolio $\bl p \in \Delta(A^+)$, Nature (i.e., a third-party) flips coins according to $\bl p$ to realize one arm from $A^+$. The reward at time $t$, which we denote by $r^t$, is defined as the expected value over the coin flips and the randomness of the rewards. Formally, $r^t=\sum_{a_i \in A^+} \bl p^t(a_i)\E\left[X(a_i)\mid \mI_t\right]$, where $\bl p^t$ is the portfolio DM selects at time $t$. 

In this paper, we limit DM to use \textit{Bayesian-safe portfolios}. A portfolio $\bl p^t$ is Bayesian-safe (or simply \textit{safe}) at time $t$ if, given the information DM has at time $t$, its expected reward is greater or equal to zero (which is the reward of the safe arm). Formally,
\begin{definition}[Bayesian-Safe Portfolio]\label{def:bayesian safety}
A portfolio $\bl p$ is safe w.r.t. $\mI$ if {\setmuskip{\thickmuskip}{0mu}$\sum_{a_i\in A}\bl p (a_i)\E\left[X(a_i)\mid \mI\right]\geq 0.$} 
\end{definition}
By restricting DM to play safe portfolios, we assure that every agent (i.e., every round) will receive at least the reward of the safe arm in expectation, independently of the other agents.\footnote{Notice that DM can explore arms with a negative expected value. However, such arms must be balanced with other arms (with positive rewards). Forbidding DM to explore arms with negative expected value trivializes the problem, as it prevents DM from exploring all arms with negative expectation. Our Bayesian safety is a compromise: it allows learning about a priori ``bad'' arms, but only when they are mixed with ``good'' arms.}
We denote the \textit{utility} achieved by an algorithm $\ALG$ by $\mU_T(\ALG)=\E\left[\sum_{t=1}^T r^t\right]$. DM wished to maximize her utility--- the sum of rewards, subject to selecting safe portfolios in every round. 

To conclude, we represent an instance of the $\ise$ problem by the tuple $\tupbracket{K, A, (X(a_i))_i, (\mu(a_i))_i}$. Notice that the horizon $T$ is not part of the description, as we often discuss a particular instance with varying $T$. When the instance is known from the context, we denote the highest possible utility of any algorithm by $\OPT_T$, where the subscript emphasizes the dependency on the number of rounds $T$. As it will become apparent later on, it is convenient to distinguish arms with positive expected rewards and arms with negative expected rewards.\footnote{We assume for simplicity that there are no arms with an expected reward of 0. Our results hold in this case as well with minor modifications.} To that end, we let 
$\above(A)\defeq\{a_i\in A:\mu(a_i) > 0 \}$. Analogously, $\below(A)\defeq\{a_i\in A:\mu(a_i)< 0 \}$. The terms above and below refer to the safe arm---an expected value of 0, which is our benchmark for safety.

Before we go on, we illustrate our setting and notation with an example.

\begin{example}\label{example with normal}
Let $K=4$, and let{\setmuskip{\thickmuskip}{0mu}
$X(a_1)\sim N(2,1)$, $X(a_2)\sim N(1,1)$, $X(a_3)\sim N(-1,1)$, and $X(a_4)\sim N(-2,1)$, where $N(\mu,\sigma^2)$ denotes the normal distribution with a mean $\mu$ and a variance $\sigma^2$. 
}%
Notice that $\above(A)=\{a_1,a_2\}$ and  $\below(A)=\{a_3,a_4\}$. In the first round, DM can play, e.g., a portfolio that comprises $a_1$ w.p. 1. She can also play the portfolio that mixes $a_1$ w.p. $\frac{1}{2}$ with $a_2$ w.p. $\frac{1}{2}$, and infinitely many other safe portfolios.

To demonstrate the technical difficulty of maximizing utility, assume that $X(a_1),X(a_2) <0$ while $X(a_3),X(a_4)>0$ (this information is \textit{not} available to DM). Indeed, this happens with positive non-negligable probability under the distributional assumptions of this example. Consider the case where DM plays the portfolio $a_1$ w.p. 1 in the first round, and $a_2$ w.p. 1 in the second round. Notice that these are safe: Given her knowledge in the first round, playing $a_1$ w.p. 1 is safe, and regardless of $x(a_1)$, playing $a_2$ w.p. 1 in the second round is also safe. After observing that $a_1$ and $a_2$ have negative rewards (since we momentarily assume $X(a_1),X(a_2) <0$), her only safe portfolio for the subsequent rounds is the safe arm $a_0$ w.p. 1.
However, there are much better ways to act in the first round. Arms with positive expected rewards are essential for exploring arms with negative expected rewards. Consider $\bl p^1$ such that $\bl p^1(a_1)=\frac{1}{3},\bl p^1(a_3)=\frac{2}{3}$ and $\bl p^1(a_2)=p^1(a_4)=0$. Notice that $\bl p^1$ is safe w.r.t. $\mI_1$, since $\bl p^1(a_1)\mu(a_1)+\bl p^1(a_3)\mu(a_3)=0$. If DM plays $\bl p^1$ in the first round, she would discover the positive reward of arm $a_3$ already in the first round w.p. $\frac{2}{3}$. It is thus immediate to see that playing $\bl p^1$ in the first round results in higher utility than playing $a_1$ w.p. 1.  This example illustrates that DM should use $a_1$ and $a_2$ as her ``budget'', and leverage their a priori positive rewards to explore the a priori inferior arms $a_3$ and $a_4$. 
\end{example}

\paragraph{Stochastic Dominance}
Our main technical contribution assumes that some of the rewards are stochastically ordered. We say that a random variable $X$ stochastically dominates (or, has first-order stochastic dominance over) a random variable $Y$ if for every $x\in (-\infty ,\infty ),$ $\Pr(X\geq x)\geq \Pr(Y\geq x)$. Notice that this dominance immediately implies that $\E\left[X\right] \geq \E\left[Y\right]$.
\begin{assumption}\label{assumption:dominance}
The reward of the arms in $\below(A)$ are stochastically ordered.
\end{assumption}
Notice that this assumption applies only to a priori inferior arms, i.e., the arms in $\below(A)$, and not to all arms. There are many natural cases for this assumption: Unit-variance (or any fixed variance) Gaussian like in Example~\ref{example with normal}, Bernoulli, log-normal and truncated normal, etc. When a formal statement relies on Assumption \ref{assumption:dominance}, we mention it explicitly. 



\section{Our Approach: A Goal Markov Decision Process for the Infinite case}\label{sec:infinite}
In this section, we assume that $T$ is infinity, that is, we care about $\lim_{T\rightarrow \infty} \mU_T$. We develop our main technical contribution, which we outline in Algorithm~\ref{alg:alg of pi}. Algorithm~\ref{alg:alg of pi} is asymptotically optimal, and can be computed in linearithmic time (sorting $\below(A)$ according to expected values). Achieving asymptotic optimality is non-trivial and requires careful planning of the portfolios we use in every round. Mixing the wrong arms even once can be detrimental, as it may prevent us from exploring more arms or arms with better rewards. Algorithm \ref{alg:alg of pi} reveals an interesting insight: as long as Assumption~\ref{assumption:dominance} holds, we can follow a neat rule for picking portfolios efficiently in an optimal manner.\footnote{The proof of Theorem~\ref{thm:optimal policy} suggests an immediate asymptotically optimal algorithm when Assumption~\ref{assumption:dominance} does not hold; however, it runs in exponential time in $K$. See the proof sketch for more details.} In the next section, we leverage our results to derive convergence bounds for the finite case as well.

When there are infinite rounds, the following interesting property occurs.
\begin{observation}\label{obs:eventually will explore}
Assume that throughout the course of execution we discovered that $X(a_i)>0$ for some arm $a_i\in A$. Then, we can use $a_i$ to explore all other arms in finite time.
\end{observation}
To illustrate, recall Example \ref{example with normal}. Assume that we are in the third round, have already explored $a_1$ and $a_2$, and discovered $X(a_1)=x(a_1)>0$ and that $X(a_2)<0$. Further, we did not explore $a_3$ and $a_4$ yet (one of which can still hide the highest payoff among all arms). Consider the portfolio 
\begin{align*}
\bl p(a) =
\begin{cases}
\frac{-\mu(a_3)}{x(a_1)-\mu(a_3)} & \textnormal{if } a=a_1\\
\frac{x(a_1)}{x(a_1)-\mu(a_3)} & \textnormal{if } a=a_3\\
0 & \textnormal{otherwise}
\end{cases}
\end{align*}
It can be verified that $\sum_{a_i \in A^+} \bl p(a_i)\E\left[X(a_i)\mid \mI_3\right]=0$, and thus $\bl p$ is safe. If we pick this portfolio, Nature flips coins to pick either $a_1$ and $a_3$, where the latter is picked with positive probability. This Bernoulli trial might end up with selecting $a_1$, but we can repeat it until Nature picks $a_3$. The number of rounds required to explore $a_3$ follows the Geometric distribution with probability of success of $\nicefrac{x(a_1)}{x(a_1)-\mu(a_3)}$ in each Bernoulli trial. Indeed, by executing Bernoulli trials until the first success, we guarantee exploring $a_3$. Zooming out of Example~\ref{example with normal} to the general case, Observation~\ref{obs:eventually will explore} suggests that once a positive reward is realized, we can execute sufficiently many Bernoulli trials to explore all the arms in finite time.

Observation \ref{obs:eventually will explore} calls for a modelling that abstracts the setting once a positive reward is realized. More particularly, we can focus on the case where all unexplored arms are revealed instantly after we realized a positive reward. To that end, we propose a constrained Goal Markov Decision Process (GMDP), which we present in the next subsection. Our goal is to find the optimal policy\footnote{We keep the term algorithm for solutions to $\ise$, and use the term policy for solutions of the GMDP.} for this GMDP, and then translate it to an asymptotically optimal algorithm for the corresponding $\ise$ instance.

\subsection{An Auxiliary Goal Markov Decision Process}\label{subsec:aux GMDP}
We construct the GMDP as follows:
\begin{itemize}
    \item Every state is characterized by the set of unobserved arms $s \subseteq A$, and we denote the set of all states by $\mS=2^A$. The initial state is $s_0=A$. 
    \item In every state $s$, we can pick any portfolio from the set of safe portfolios w.r.t. the prior information, i.e., from
    $
    \safe(s)=\left\{\bl p \in \Delta(s) : \sum_{a\in s}\bl p(a)\mu(a) \geq 0 \right\},
    $
    where $\Delta(s)$ is the set of all distributions over the elements of $s$.
    \item Given a portfolio $\bl p$, an arm index is sampled. If the realized arm is $a$, the GMDP transitions to state $s\setminus \{ a\}$. If $\safe(s)$ is empty, then we say that $s$ is \textit{terminal}. In particular, we can reach a terminal state if we ran out of arms with a positive expected value, i.e., $\above(s)=\emptyset$, or if we have explored all arms, i.e., $s=\emptyset$.
    \item Rewards are obtained in terminal states solely. The reward of a terminal state $s$ is
\[
R(s) \defeq
\begin{cases}
\max_{a\in A} X(a) & \textnormal{if $X(a') >0$ for $a'\in A\setminus s$} \\
0 & \textnormal{otherwise}
\end{cases}.
\]
\end{itemize}
Notice that the reward of a terminal state $s$ depends on the rewards of the arms not in $s$. If we reached $s$, then we have explored $A\setminus s$. Following our intuition of Bernoulli trials, if at least one of $(X(a))_{a\in A\setminus s}$ is positive, the reward we get is $\max_{a\in A} X(a)$.    Notice that the GMDP allows for the realization of each arm only once: in a state $s$, if $a$ is realized we transition to $s\setminus\{a\}$. This ensures that the profiles in $\safe(s)$ are also safe for an $\ise$ algorithm that uses a policy for the GDMP, provided that it did not explore the arms in $s$.

Let $W(\pi,A)$ be the reward of a policy $\pi$ starting at the initial state $s_0=A$. Moreover, let $W^*(A)$ denote the highest possible reward of any policy, i.e., $W^*(A)=\sup_{\pi}W(\pi,A)$. Due to the construction, 
\begin{observation}\label{obs: U leq W*}
For any algorithm $\ALG$, it holds that $\lim_{T \rightarrow \infty }\mU_T(\ALG) \leq W^*(A)$.
\end{observation}
Furthermore, given any policy $\pi$ for the GMDP, we can construct an algorithm $\ALG(\pi)$ for the corresponding $\ise$ instance. At the beginning, $\ALG(\pi)$ picks portfolios according to $\pi$ and updates the state. Then, once $\pi$ reaches a terminal state,  $\ALG(\pi)$ either executes Bernoulli trials until full exploration as we discussed before (in case a positive reward was realized), or plays the safe arm forever (we further elaborate when introducing Algorithm~\ref{alg:alg of pi}). It is straightforward to see that
\begin{observation}\label{obs: U get W}
for any policy $\pi$, there exists an algorithm $\ALG(\pi)$ such that $\lim_{T \rightarrow \infty }\mU_T(\ALG(\pi)) = W(\pi,A)$.
\end{observation}
As an immediate corollary, if $\pi$ is an optimal policy, namely, $W(\pi,A)=W^*(A)$, then $\ALG(\pi)$ is asymptotically optimal, i.e., $\lim_{T\rightarrow \infty}\mU_T(\ALG(\pi)) \geq \lim_{T\rightarrow \infty}\OPT_T$. We can thus focus on finding an optimal GMDP policy, and later translate it to an asymptotically optimal algorithm for the $\ise$ problem.

\subsection{Optimal GMDP Policy}\label{subsec:optimal GMDP policy}
\begin{figure}[tb]
\begin{minipage}{.48\linewidth}
\begin{algorithmpolicy}[H]
\caption{Optimal GMDP Policy ($\OGP$)\label{policy:pi star}}
\begin{algorithmic}[1]
\REQUIRE a state $s\subseteq A$.
\ENSURE a portfolio from $(\bl p_{i,j})_{i,j}$ or $\emptyset$.
\IF{$s$ is terminal \label{policy:if terminal}} {
\STATE \textbf{return} $\emptyset$. \label{policy:return empty}
}
\ELSE\label{policy:non terminal}{
\STATE pick any arbitrary $a_i \in \above(s)$.\label{policy:pick arbitrary}
\IF{$\below(s)=\emptyset$\label{policy:if no below}} {
		\RETURN $\bl p_{i,i}$. \label{policy:return double above}
}
\ELSE\label{policy:if has below}{
\STATE pick $a_{j^*}\in \argmax_{a_j \in \below(s)} \mu(a_j)$.\label{policy:pick below}
		\RETURN $\bl p_{i,j^*}$.\label{policy:return mix}
}
\ENDIF
}
\ENDIF
\end{algorithmic}
\end{algorithmpolicy}
\end{minipage}
~~~
\begin{minipage}{.48\linewidth}

\begin{figure}[H]
    \includegraphics[scale=1.]{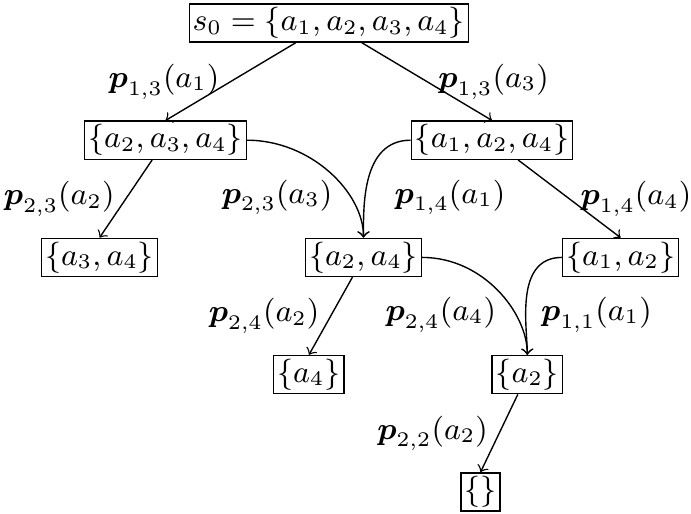}
    \small
    \caption{Illustration of $\OGP$ for Example~\ref{example with normal}.\label{fig:tree small}}
\end{figure}

\end{minipage}
\end{figure}
Next, we devise an optimal policy for this GMDP. To that end, we highlight the following family of portfolios that mix at most two arms. Consider a pair of arms, $a_i\in \above(A)$, and $a_j\in \below(A)$. The portfolio
\begin{align}\label{eq:blp from body}
\bl p_{i,j}(a) \defeq 
\begin{cases}
\frac{-\mu(a_j)}{\mu(a_i)-\mu(a_j)} & \textnormal{if } a=a_i\\
\frac{\mu(a_i)}{\mu(a_i)-\mu(a_j)} & \textnormal{if } a=a_j\\
0 & \textnormal{otherwise}
\end{cases}.
\end{align}
mixes $a_i$ and $a_j$ while maximizing the probability of exploring $a_j$ (the a priori inferior arm). The reader can verify that $\bl p_{i,j}$ is indeed safe, yielding an expected value of precisely zero. For completeness, for every $a_i \in \above(A)$, we also define $\bl p_{i,i}$ as a deterministic selection of $a_i$, e.g., $\bl p_{i,i}(a)=1$ if $a=a_i$, and zero otherwise.

We are ready to present the optimal GMDP policy (hereinafter $\OGP$ for shorthand), which we formalize via Policy~\ref{policy:pi star}. Given a state $s$, $\OGP(s)$ operates as follows. If $s$ is terminal, it returns the empty set (Lines~\ref{policy:if terminal}--\ref{policy:return empty}). Otherwise, if $s$ is non-terminal, we enter the ``else'' clause in Line~\ref{policy:non terminal}, and pick any arbitrary arm $a_i$ from $\above(s)$. Then, we have two cases. If $\below(s)$ is empty, $\OGP$ returns $\bl p_{i,i}$ (Line~\ref{policy:return double above}). Else, it picks the best arm from $\below(s)$ in terms of expected reward, which we denote by $a_{j^*}$ (Line~\ref{policy:pick below}), and returns $\bl p_{i,j^*}$ (Line~\ref{policy:return mix}). Overall, $\OGP$ returns $\emptyset$ if $s$ is terminal or a portfolio from $(\bl p_{i,j})_{i,j}$ if $s$ is non-terminal.
\begin{theorem}\label{thm:optimal policy}
Fix any arbitrary $\ise$ instance satisfying Assumption~\ref{assumption:dominance}. For $\OGP$ prescribed in Policy~\ref{policy:pi star}, it holds that $W(\OGP,A)=W^*(A)$.
\end{theorem}
\begin{proof}[\textnormal{\textbf{Proof Overview of Theorem \ref{thm:optimal policy}}}]
Proving Theorem \ref{thm:optimal policy} is the main technical achievement of this paper. We defer it to {\ifnum\Includeappendix=0{the appendix}\else{Section~\ref{sec:thm1 outline}}\fi} and  outline it below.

First, using canonical arguments, we show that there exists a \textit{stationary} optimal policy (i.e., a policy that picks a portfolio in the current state independently of the states that led to it). We call a stationary policy that picks portfolios from $(\bl p_{i,j})_{i,j}$ a $\mP$-valid policy. With a bit more theory work, we show that there exists an optimal policy that is $\mP$-valid. We can find such an optimal policy inefficiently using dynamic programming: in every state $s$, assume we know the optimal solution for $s' \subset s$. Since $s$ leads to states of the form $s/{a}$ for $a\in s$, we can compute the expected reward for each portfolio $(\bl p_{i,j})_{i,j}$, and pick the best one. There are $2^K$ states and $O(K^2)$ portfolios from $(\bl p_{i,j})_{i,j}$ we can play in every state, each taking $O(K)$ computations to assess; therefore, this by itself guarantees finding an optimal policy in time $O(2^K K^3)$.

To substantially reduce the computation, we need to understand the crux of the dynamic programming. A crucial ingredient of the procedure is the probability of reaching the empty state, representing the case in which we explored all arms. For every state $s\in \mS$ and a policy $\pi$, we denote by $Q(\pi,s)$ the probability starting at $s\subseteq A$, following the policy $\pi$ and reaching the empty state (exploring all arms). 
We then prove a rather surprising feature of $Q$: $Q$ is policy independent. Namely, we show that for every two $\mP$-valid policies $\pi,\rho$ and every state $s\in \mS$, it holds that $Q(\pi,s)=Q(\rho,s)$. Interestingly, the arguments we prove until this point (including the existence of an optimal $\mP$-valid policy and the above $Q$ equivalence) do not rely on Assumption \ref{assumption:dominance}.

Equipped with these tools, we make the final arguments. We use the recursive structure of $Q$ and the reward function $W$ to prove the optimality of $\OGP$ inductively over the sizes of $\above(s)$ and $\below(s)$. This step makes use of further insights about the dynamic programming, like monotonicity in the number of available arms.
\end{proof}
We exemplify $\OGP$ for the induced GMDP of Example~\ref{example with normal} in Figure \ref{fig:tree small}. Every node is a state---the root is $s_0=A$. In $s_0$, we play the portfolio $\bl p_{1,3}$, which mixes $a_1\in \above(A)$ and $a_3 \in \below(A)$ (we could have equivalently picked $\bl p_{2,3}$, since Line~\ref{policy:pick arbitrary} provides us freedom in picking $a_i\in\above(A)$). The split follows from Nature's coin flips: left if the realized action is $a_1$ (w.p. $\bl p_{1,3}(a_1)$), and right if the action is $a_3$ (w.p. $\bl p_{1,3}(a_3)$). The leaves $\{a_3,a_4\}$, $\{a_4\}$, and $\{\}$ are the terminal states. 

Notice that Observation \ref{obs: U get W} and Theorem \ref{thm:optimal policy} together suggest an optimal algorithm for $\ise$ with an infinite horizon. We denote this algorithm by $\SEGB$ for brevity (following the notation of Observation \ref{obs: U get W}, $\SEGB=\ALG(\OGP)$). $\SEGB$ picks portfolios according to $\OGP$ until a positive reward is realized (Lines~\ref{algpi:while}--\ref{algpi:update s}). Then, if $\SEGB$ realized a positive reward (Line~\ref{algpi:pos then geo}), it explores the remaining arms in $\above(A)$ (Line~\ref{algpi:explore remaining}), executes finitely many Bernoulli trials to explore all arms (Line~\ref{algpi:geo}), and exploits the best arm (Line~\ref{algpi:exploit best}). Otherwise, it plays the safe arm $a_0$ forever. As an immediate corollary of Observations \ref{obs: U leq W*}--\ref{obs: U get W} and Theorem \ref{thm:optimal policy}:
\begin{corollary}
Fix any arbitrary $\ise$ instance satisfying Assumption~\ref{assumption:dominance}. It holds that $\lim\limits_{T \rightarrow \infty }\mU_T(\SEGB) = \lim\limits_{T \rightarrow \infty }\OPT_T$.
\end{corollary}

\begin{algorithm}[t]
\renewcommand{\algorithmiccomment}[1]{\texttt{\kibitz{blue}{\##1}}}

\caption{Safe Exploration via GMDP and Bernoulli Trials ($\SEGB$) \label{alg:alg of pi}}
\begin{algorithmic}[1]
\STATE $s\gets A$
\WHILE[$s$ is not a terminal state] {$\OGP(s)\neq \emptyset$\label{algpi:while}}{
\STATE play $\OGP(s)$, and denote the realized action by $a_k$.\label{algpi:play with ogp}
\IF[a positive reward was realized]{$x_{a_k}>0$} {
		\STATE \textbf{break}. 
}
\ENDIF
\STATE $s\gets s\setminus \{a_k\}$.\label{algpi:update s}
}
\ENDWHILE
\IF{$x_{a_k}>0$ for some $a_k\in A$\label{algpi:pos then geo}} {
        \STATE explore all the remaining arms in $\above(A)$.\label{algpi:explore remaining}
        \STATE $a_{k^*} \gets \argmax_{a_i} x(a_i)$. \COMMENT{best among all the explored arms} 
		\STATE execute Bernoulli trials mixing $a_{k^*}$ with every other unexplored arm until all are revealed. \label{algpi:geo}
		\STATE play the best arm forever. \label{algpi:exploit best}
}
\ENDIF
\STATE \textbf{else}, play the safe arm $a_0$ forever. \label{algpi:exploit a0}
\end{algorithmic}
\end{algorithm}

\section{Instance-Dependent Convergence Rate}\label{sec:finite}
In this section, we derive convergence rates for $\SEGB$. First, we note that in certain cases, $\SEGB$ can be slightly modified to be optimal even for finite $T$. To illustrate, consider $(X(a_i))_i$ that are supported on $\{-1,+1\}$. Revealing a positive reward, which is 1 due to the $\{-1,+1\}$ support, suggests that we need not explore any further (since the other rewards cannot outperform $+1$); thus, the Bernoulli trials in Line~\ref{algpi:geo} are redundant. Formalizing this intuition,
\begin{proposition}\label{prop:bernoulli opt}
\omer{for the camera-ready and arxiv, use the commented one}
Fix any arbitrary $\ise$ instance such that $(X(a_i))_i \in \{x^-,x^+\}$ w.p. 1, i.e., the rewards take only two possible values almost surely. Let $\SEGB'$ be a modified version of $\SEGB$ that exploits once a reward of $x^+$ is realized. Then, there exists $T_0$ such that whenever $T\geq T_0$, $\mU_T(\SEGB') =\OPT_T$.
\end{proposition}
\omer{for this to occur, $\SEGB'$ must use a policy that orders $\above(A)$ according to the stochastic order. And, it should stop using $\OGP$ once it realizes a positive arm.}
However, $\SEGB$ is sub-optimal in the general finite case. This sub-optimality comes from various factors. First, the Bernoulli trials in Line~\ref{algpi:geo} might be worthless in expectation. To illustrate, assume that there are a few rounds left, and one unexplored arm from $\below(A)$. In such a case, the one-time cost of exploring this arm can be significantly lower than the value it brings to the remaining rounds. Second, recall that when we get to Line~\ref{algpi:explore remaining}, some of the arms in $\above(A)$ might still be explored. Deciding whether to mix such arms in Bernoulli trials or exploit them constitutes another barrier. We leave this technical challenge for future work. Instead, our goal is to derive instance-dependent bounds for Algorithm~\ref{alg:alg of pi}. 

Given an $\ise$ instance, let $\delta_1 = \min_{a_i\in \below(A): X(a_i)>0} X(a_i)$. Namely, $\delta_1$ is the lowest positive reward among the arms of $\below(A)$.\footnote{We refer to the \textit{realized} values. Our analysis has the same spirit as  expressing the regret in terms of the actual gaps in MABs \cite{agrawal2012analysis}.} In addition, let $\delta_2 = \max_{a_i\in \above(A): X(a_i)>0} X(a_i)$ denote the highest positive reward among the arms of $\above(A)$. Recall that in Line~\ref{algpi:geo} of Algorithm~\ref{alg:alg of pi}, we execute Bernoulli trials with the best seen arm. The lower bound on the reward of that arm is $\delta=\max\{\delta_1,\delta_2\}$.\footnote{If $\delta_1$ and $\delta_2$ are not defined, both $\SEGB$ and $ \OPT_T$ use the safe arm $a_0$ after at most $K$ rounds.}  Furthermore, let $\gamma = \max_{a_i\in A:\mu(a_i)<0}\abs{\mu(a_i)}$. Notice that $\gamma$ quantifies the highest absolute value of an arm in $\below(A)$. The success probability of each Bernoulli trial (which leads to exploring one additional arm) is lower bounded by 
$ \frac{\delta}{\delta+\gamma}$. Consequently, after $K(1+\frac{\gamma}{\delta})$ Bernoulli trials in expectation, we would explore all arms. We formalize this intuition via Proposition~\ref{prop:i-d bounds} below.
\begin{proposition}\label{prop:i-d bounds}
Fix any arbitrary $\ise$ instance satisfying Assumption~\ref{assumption:dominance}, and let $\delta>0$ and $\gamma$ be the quantities defined above. There exists $T_0$ such that whenever $T>T_0$, it holds that
$
\mU_T(\SEGB) \geq \left(  1-\frac{K(1+\frac{\gamma}{\delta})}{T}\right) \OPT_T.
$
\end{proposition}
For example, assume that $(X(a_i))_i$ are arbitrarily distributed in the discrete set $\{-H,\dots,0,\dots H\}$. In such a case, $\frac{\gamma}{\delta} \leq H$; thus, $\SEGB$ is optimal up to a multiplicative factor of $\frac{K(H+1)}{T}$ (or an additive factor of $\frac{KH(H+1)}{T}$).

\section{Conclusion and Open Problems}\label{sec:discussion}
We presented a model that is inspired by recent work on safe reinforcement learning. In our model, safety constraints apply to \textit{every} round. Moreover, reaching optimality relies heavily on careful planning, which is one more crucial element of safety-critical systems.

We see considerable scope for future work. On the technical side, one possible follow-up is to relax Assumption~\ref{assumption:dominance}. Once we relax it, $\OGP$ ceases to be optimal due to incorrect planning. We conjecture that a different Gittins index-like policy is optimal; see~{\ifnum\Includeappendix=0{the appendix}\else{Proposition \ref{prop:index with ugeq one} in the appendix}\fi} for more details. A second possible follow-up is to consider stochastic bandits with Bayesian priors. Our Definition~\ref{def:bayesian safety} is readily extendable to this case as well. 

Conceptually, our definition of safety concerns the expected value, yet in many settings other factors should be taken into account, e.g., variance. In such a case, our constraint can be generalized to $\sum_{a_i\in A}\bl p (a_i)\E\left[f(X(a_i))\mid \mI\right]\geq 0$, for some function $f$. This more general formulation can express risk-aversion or risk-seeking behavior, as well as more complicated quantities that depend on the reward distribution. We note that our results from Section~\ref{sec:infinite} hold for any general $f$, as long as it agrees with the stochastic order. Finally, our notion of invariable safety could be applied to other explore--exploit models, e.g., to MDPs.

{\ifnum\Putacknowledgement=1{
\section*{Acknowledgements}
The work of G. Bahar, O. Ben-Porat and M. Tennenholtz is funded by the European Research Council (ERC) under the European Union's Horizon 2020 research and innovation programme (grant agreement n$\degree$  740435). The work of K. Leyton-Brown is funded by the NSERC Discovery Grants program, DND/NSERC Discovery Grant Supplement, Facebook Research and Canada CIFAR AI Chair Amii. Part of this work was done while K. Leyton-Brown was a visiting researcher at Technion - Israeli Institute of Science and was partially funded by the European Union's Horizon 2020 research and innovation programme (grant agreement n$\degree$  740435).
}\fi}


\input{bibdb.bbl}
{\ifnum\Includeappendix=1{ 
\appendix

\section{Proof Outline for Theorem \ref{thm:optimal policy}}\label{sec:thm1 outline}
\input{input/preliminaries_for_Theorem_XX}

\section{Proof of Lemma \ref{lemma:equivalence}}\label{sec:proof of lemma}
\input{input/safe_exploration_lemma_proof}

\section{Additional Statements for Lemma \ref{lemma:equivalence}}
\input{input/safe_exploration_lemma_additions}

\section{Base Cases for Lemma \ref{lemma:equivalence}}\label{sec:base for lemma}
\input{input/safe_exploration_lemma_base}

\section{Proof of Theorem \ref{thm:holy grail}}\label{sec:proof of thm}
\input{input/safe_exploration_theorem}

\section{Statements for Theorem \ref{thm:holy grail}}\label{sec:for theorem}
\input{input/safe_exploration_thm_additions}
}

\section{Proofs of Observations and Propositions from Sections \ref{sec:infinite} and \ref{sec:finite}}\label{sec:appendix main body}
\input{input/body-proofs}

\section{Proof of Statements from Section \ref{sec:thm1 outline}}\label{sec:aux}
\input{input/proposition_p_valid}
\input{input/auxiliary_statements}

\fi}

\end{document}

%% file: defs.tex
\usepackage{gensymb} 
\usepackage{dsfont}
\usepackage{mathtools}
\usepackage{amsthm}

\makeatletter
\let\original@algocf@latexcaption\algocf@latexcaption
\long\def\algocf@latexcaption#1[#2]{%
  \@ifundefined{NR@gettitle}{%
    \def\@currentlabelname{#2}%
  }{%
    \NR@gettitle{#2}%
  }%
  \original@algocf@latexcaption{#1}[{#2}]%
}
\makeatother

\newcommand{\ug}{U_>}
\newcommand{\ul}{U_<}

\newcommand{\suff}{\textnormal{suffix}}

\newcommand{\prefix}{\textit{prefix}}
\newcommand{\safe}{\textit{safe}}
\newcommand{\pathto}[2]{#1\stackrel{\pi}{\rightsquigarrow}#2}
\newcommand{\pth}[3]{#1\stackrel{#2}{\rightsquigarrow}#3}
\newcommand{\pathtorho}[2]{#1\stackrel{\rho}{\rightsquigarrow}#2}

\renewcommand{\above}{\textnormal{above}}
\newcommand{\below}{\textnormal{below}}

\renewcommand{\refeq}[1]{(\ref{#1})}
\newcommand{\tmu}[1]{\tilde \mu_{#1}}
\newcommand{\sigr}{\sigma^{\textnormal{right}}}
\newcommand{\sigl}{\sigma^{\textnormal{left}}}
\newcommand{\setmuskip}[2]{#1=#2\relax}

\newenvironment{proofof}[1]{\begin{proof}[\textnormal{\textbf{Proof of #1}}]}{\end{proof}} 

\newcount\Comments  
\newcount\Includeappendix 
\newcount\Putacknowledgement 

\definecolor{darkgreen}{rgb}{0,0.6,0}
\newcommand{\kibitz}[2]{\ifnum\Comments=1{\color{#1}{#2}}\fi}

\newcommand{\omer}[1]{\kibitz{blue}{[Omer:#1]}}

\makeatletter
\newcommand{\RemoveAlgoNumber}{\renewcommand{\fnum@algocf}{\AlCapSty{\AlCapFnt\algorithmcfname}}}
\newcommand{\RevertAlgoNumber}{\algocf@resetfnum}
\makeatother

\newcommand{\mP}{\mathcal{P}}

\newcommand{\mI}{\mathcal{I}}

\newcommand{\mH}{\mathcal{H}}

\newcommand{\mS}{\mathcal{S}}

\newcommand{\mU}{\mathcal{U}}
\newcommand{\ind}{\mathds{1}}

\newcommand{\defeq}{\stackrel{\text{def}}{=}}

\newcommand{\tupbracket}[1]{\left\langle {#1} \right\rangle}

\newtheorem{claim}{Claim}
\newtheorem{assumption}{Assumption}
\newtheorem{definition}{Definition}
\newtheorem{theorem}{Theorem}
\newtheorem{proposition}{Proposition}
\newtheorem{lemma}{Lemma}
\newtheorem{corollary}{Corollary}
\newtheorem{observation}{Observation}

\theoremstyle{definition}

\newtheorem{example}{Example}

\newcommand\bl[1]{\boldsymbol{ #1 } }

\newcommand\abs[1]{\left| #1  \right|}
\DeclareMathOperator*{\argmin}{arg\,min} 
\DeclareMathOperator*{\argmax}{arg\,max} 

\DeclareMathOperator{\E}{\mathbb{E}}

\newcounter{algorithmpolicy}
\makeatletter
\newenvironment{algorithmpolicy}[1][htb]{%
  \let\c@algorithm\c@algorithmpolicy
    \renewcommand{\ALG@name}{Policy}
   \begin{algorithm}[#1]%
  }{\end{algorithm}}
\makeatother

%% file: input/preliminaries_for_Theorem_XX.tex
In this section, we outline the proof of Theorem~\ref{thm:optimal policy}. We begin with several notations and definitions we use extensively in the proof.
\subsection{Preliminaries}
We denote the set of all states by $\mS=2^A$. A \textit{policy} is a mapping from previous states and actions to a randomized action. Formally, let $\mH$ be the set of histories, $\mH=\cup_{k=0}^K\left(\mS \times \Delta(A)\right)^k$ be a tuple of pairs of states and randomized action taken. A policy $\pi$ is a function $\pi:\mH \times \mS \rightarrow \Delta(A)$. We say that a policy is \textit{safe} if for every $h\in \mH,s\in \mS$, $\pi(h,s)\in \safe(s)$. From here on, we consider safe policies solely. Given a policy $\pi$ and a pair $(h,s)$, we let $W(\pi,h,s)$ denote the expected reward of $\pi$ when starting from $s$ after witnessing $h$. Namely,
\begin{align}\label{eq:W elaborated}
W(\pi,h,s) = 
\begin{cases}
R(s) & \textnormal{if }\safe(s)=\emptyset\\
\sum_{a\in s}\pi(h,s)(a)W(\pi,h\oplus(s,a),s\setminus \{a\}) & \textnormal{otherwise}
\end{cases}.
\end{align}
For every state $s$, let $W^*(s)=\sup_{\pi'}W(\pi ',\emptyset,s)$. \footnote{As we show in Proposition \ref{prop:optimal p valid}, there exists a policy that attains this supremum.}While policies may depend on histories, it is often suffice to consider stationary policies.
\begin{definition}[Stationary]
A safe policy $\pi$ is \textit{stationary} if for every two histories $h,h' \in \mH$ and a state $s \in \mS$, $\pi(h,s)=\pi(h',s)$.
\end{definition}
As we show later in Proposition \ref{prop:optimal p valid}, there exists an optimal stationary policy; hence, from here on we address stationary policies solely. When discussing stationary policies, we thus neglect the dependency on $h$, writing $\pi(s)$. For stationary policies, the definition of $W$ is much more intuitive: Given a stationary policy $\pi$ and a state $s$,
\begin{align}\label{eq:w with terminal}
W(\pi,s)=\sum_{\substack{s_t\in \mS:\\ s_t \textnormal{ is terminal}}}\Pr(\pathto{s}{s_t})R(s_t),
\end{align}
where $\pathto{s}{s_t}$ indicates the event that, starting from $s$ and following the actions of $\pi$, the GMDP terminates at $s_t$. 

An additional useful notation is the following. For every state $s\in \mS$, we denote by $Q(\pi,s)$ the probability starting at $s\subseteq A$, following the policy $\pi$ and exploring all arms. Formally,
\[
Q(\pi,s)= \Pr(\pathto{s}{\emptyset}),
\]
Note that $Q$ is defined recursively: Namely, if $\pi(s)=\bl p_{i,j}$ for a non-terminal state $s$, then
\[
Q(\pi,s)=\bl p_{i,j}(a_i)Q(\pi,s\setminus\{a_i\})+\bl p_{i,j}(a_j)Q(\pi,s\setminus\{a_j\}).
\]
It will sometimes be convenient to denote $Q(\pi,\above(s),\below(s))$ for $Q(\pi,s)$, thereby explicitly stating the two distinguished sets of arms.

\subsection{Binary Structure}\label{subsec:bin}
A structural property of the above GMDP is that in every terminal state $s_t$, $\above(s_t)=\emptyset$, or otherwise we could explore more arms; thus, intuitively, the arms in $\above(A)$ provide us ``power'' to explore the arms of $\below(A)$. Following this logic, in every state we should aim to explore arms from $\below(A)$ and not those of $\above(A)$, subject to satisfying the safety constraint.

Recall the definition of $\bl p_{i,j}$ and $\bl p_{i,i}$ from Equation~\ref{eq:blp from body} in Subsection~\ref{subsec:optimal GMDP policy}. Next, we define $\mP,\mP'$ such that
\[
\mP\defeq \{\bl p_{i,j}\mid a_i\in \above(A), a_j\in \below(A) \}, \qquad \mP'\defeq\{\bl p_{i,i}\mid a_i\in \above(A) \}.
\]
Notice that $\mP\cup \mP'$ includes $O(K^2)$ actions, while $\safe(s)$ for a state $s$ is generally a convex polytope  with infinitely many actions. Further, in every non-terminate state $s$, ${\safe(s) \cap (\mP \cup \mP') \neq \emptyset}$.
\begin{definition}[$\mP$-valid]
A safe policy $\pi$ is $\mP$-valid if for every non-terminal state $s\in \mS$,
\begin{itemize}
\item if $\below(s)\neq \emptyset$, then $\pi(s)\in \mP$;
\item else, if $\below(s)= \emptyset$, then $\pi(s)\in \mP'$.
\end{itemize} 
\end{definition}
Observe that $\mP$ is a strict subset of all the safe actions in the state $A$, which incorporate mixes of at most two arms. However, the set of safe actions $\safe(s)$ for $s\subseteq A$ may include distributions mixing several elements of $A$. Due to the convexity of $W(\pi,s)$ in $\pi(s)$ (see the elaborated representation of $W$ in Equation \refeq{eq:W elaborated}), the GMDP exhibits a nice structural property, as captured by the following Proposition \ref{prop:optimal p valid}. 
\begin{proposition}\label{prop:optimal p valid}
There exists an optimal policy which is $\mP$-valid.
\end{proposition}
The proof of Proposition \ref{prop:optimal p valid} appears in Section~\ref{sec:aux}. 
Due to Proposition \ref{prop:optimal p valid}, we shall focus on $\mP$-valid policies. Such policies are easy to visualized using trees, as we exemplify next.\footnote{In Figure~\ref{example with normal} we illustrated the optimal policy using a graph that is not a tree. However, a tree structure  serves better the presentation of our technical statements.}
\begin{example}\label{example with four}
We reconsider Example~\ref{example with normal}, but neglect the actual distributions (as we only care about the expected values). Let $A=\{a_1,a_2,a_3,a_4\}$, with $\above(A)=\{a_1,a_2\}$ and $\below(A)=\{a_3,a_4\}$. Consider the tree description in Figure \ref{fig:tree example}. The root of the tree is the set of all arms. At the root, the policy picks $\bl p_{1,3}$. The outgoing left edge represents the case the realized action is $a_1$, which happens w.p. $\bl p_{1,3}(a_1)$. In such a case, the new state is $\{a_2,a_3,a_4\}$. With the remaining probability, $\bl p_{1,3}(a_3)$, the new state will be $\{a_1,a_2,a_4\}$. Leaves of the tree are terminal states, where no further exploration could be done. For instance, in the leftmost leaf, $\{a_3,a_4\}$, the only arms explored are $\{a_1,a_2\}$. The two highlighted nodes represent the same state. Since the presented policy is $\mP$-valid, it is stationary; hence, the policy acts exactly the same in these two nodes and their sub-trees.
\end{example}
Notice that the tree in Figure \ref{fig:tree example} represents only the \textit{on-path} states, i.e., states that are reachable with positive probability, while policies are functions from the entire space of states; thus, two different policies can be described using the same tree. Nevertheless, the tree structure is convenient and will be used extensively in our analysis. When we define a policy using a tree, we shall also describe its behavior at \textit{off-path} states.

The policy exemplified in Figure \ref{fig:tree example} has an additional combinatorial property: In every state $s$, it takes an action according to some order of the arms. This property is manifested in the following Definitions \ref{def:right ordered} and \ref{def:left ordered}.
\begin{definition}[Right-ordered policy]\label{def:right ordered}
A $\mP$-valid policy $\pi$ is right-ordered if there exist a bijection $\sigr_\pi: \below(A)\rightarrow [\abs{\below(A)}]$ such that in every state $s$ with $\below(s) \neq \emptyset$, $\pi(s)=\bl p_{i, {j^*}}$ where $a_i \in \above(s)$ and $a_{j^*} = \argmin_{a_j \in \below(s)} \sigr_\pi(a_j)$.
\end{definition}
\begin{definition}[Left-ordered policy]\label{def:left ordered}
A $\mP$-valid policy $\pi$ is left-ordered if there exist a bijection $\sigl_\pi: \above(A)\rightarrow [\abs{\above(A)}]$ such that in every state $s$ with $\below(s) \neq \emptyset$, $\pi(s)=\bl p_{{i^*}, j}$ where $a_j \in \below(s)$ and $a_{i^*} = \argmin_{a_i \in \above(s)} \sigl_\pi(a_i)$.
\end{definition}
In addition, we say that a policy is \textit{ordered} if it is right-ordered and left-ordered. To illustrate, observe the example in Figure \ref{fig:tree example}. The tree depicts an ordered policy, with $\sigl=(a_1,a_2)$ and $\sigr=(a_3,a_4)$. Notice that ordered policies are well-defined for off-path states.
\begin{figure}[t]
\centering
\includegraphics[scale=1.]{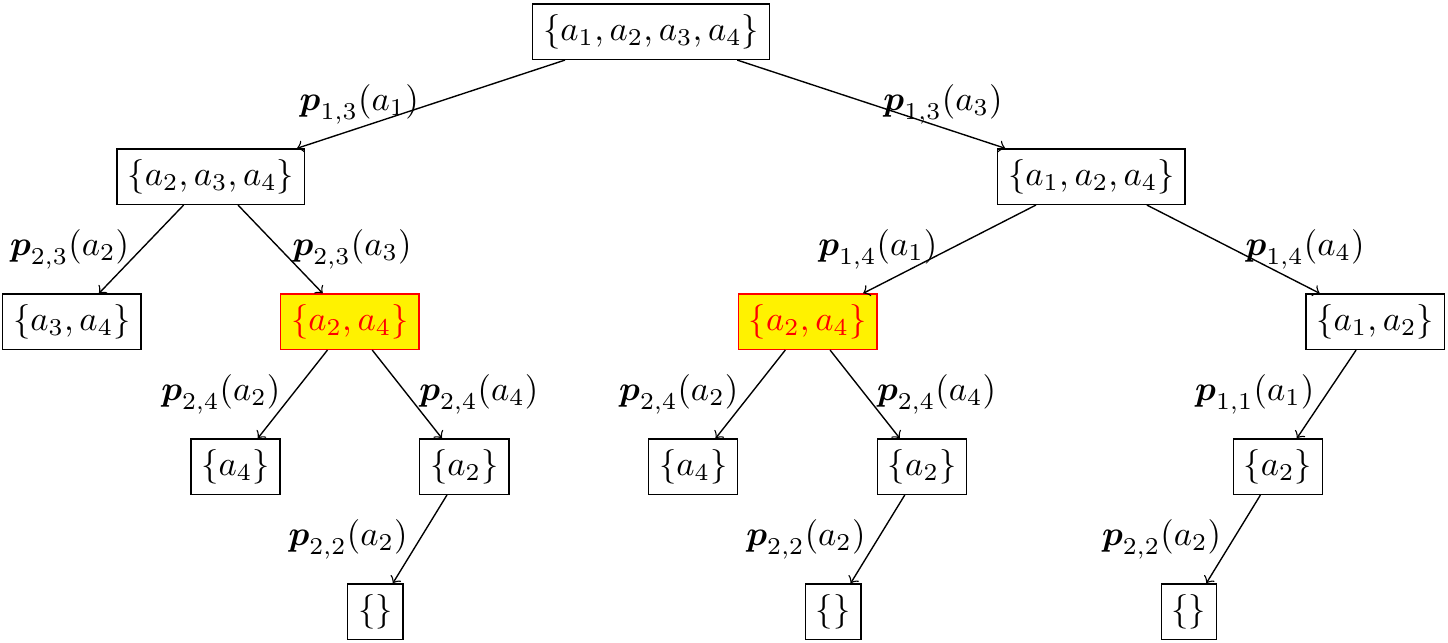}
\label{fig:tree example}
\caption{The policy described in Example \ref{example with four}.  
Every node represents a state (the mapping is onto, but not one-to-one). Outgoing left edges imply coin flips resulted in an arm from $\above(v)$, and outgoing right edges imply an arm from $\below(v)$. Leaves correspond to terminal states, where no action could be taken. 
}
\end{figure}
\subsection{Stochastic Dominance and Non-triviality}\label{subsec:stopchastic}
In this subsection, we demonstrate why the problem is still challenging even under Assumption \ref{assumption:dominance}. Recall that Proposition \ref{prop:optimal p valid} ensures that an optimal $\mP$-valid policy exists. One natural candidate for the optimal policy is ordered policy $\pi$ with with any order $\sigl_\pi$ and $\sigr_\pi$ that follow the stochastic order on $\below(A)$. Indeed, as we show formally in Theorem~\ref{thm:holy grail}, this intuition is appropriate. However, we explain shortly, the optimally of this policy cannot be shown without further work.

Consider a state $s\in\mS$, such that $\above(s),\below(s) \geq 2$. Let $a_i = \argmin_{a_{i'}\in\above(s)}\sigl_\pi(a_{i'})$, and  $a_j = \argmin_{a_{j'}\in\below(s)}\sigr_\pi(a_{j'})$. In addition, let $a_{\tilde j}\in \below(s), a_{\tilde j} \neq a_j$. The action $\bl p_{i,j}$, which mixes the minimal elements according to the stochastic order, is weakly superior to  $\bl p_{i,{\tilde{j}}}$ if 
{
\thinmuskip=.2mu
\medmuskip=0mu plus .2mu minus .2mu
\thickmuskip=1mu plus 1mu
\begin{align}\label{eq:why hard}
&\bl p_{i,j}(j)W^*(s\setminus \{a_j\})-\bl p_{i,{\tilde j}}({\tilde j})W^*(s\setminus \{a_{\tilde j}\})+\left(\bl p_{i,j}(i)-\bl p_{i,{\tilde j}}(i)\right)W^*(s\setminus\{a_i\}) \geq 0.
\end{align}}%
By our selection of $a_j,a_{\tilde j}$, we know that $\bl p_{i,j}(i)-\bl p_{i,{\tilde j}}(i) \leq 0$; hence, the third term is non-positive. Moreover, as we show in Claim \ref{claim:ass is not for W} in Section~\ref{sec:aux}, 
stochastic dominance does not imply that $W^*(s\setminus \{a_j\}) \geq W^*(s\setminus \{a_{\tilde j}\})$; thus, it is not even clear that the expression
{
\thinmuskip=.2mu
\medmuskip=0mu plus .2mu minus .2mu
\thickmuskip=1mu plus 1mu
\begin{align*}
&\bl p_{i,j}(j)W^*(s\setminus \{a_j\})-\bl p_{i,{\tilde j}}({\tilde j})W^*(s\setminus \{a_{\tilde j}\})
\end{align*}}%
which accounts for the first two terms in Inequality (\ref{eq:why hard}), is non-negative. Therefore, we cannot claim for Inequality  (\ref{eq:why hard}) without revealing the structure of $W^*$, even when Assumption \ref{assumption:dominance} holds. 
\subsection{Proof Overview}\label{subsec:results}
We are ready to prove Theorem~\ref{thm:optimal policy}. The main tool in our analysis is Lemma \ref{lemma:equivalence}. Lemma \ref{lemma:equivalence} reveals a rather surprising feature of $Q$: $Q$ is policy independent. 
\begin{lemma}[Equivalence lemma]\label{lemma:equivalence}
For every two $\mP$-valid policies $\pi,\rho$ and every state $s\in \mS$, it holds that $Q(\pi,s)=Q(\rho,s)$.
\end{lemma}
The proof of Lemma \ref{lemma:equivalence} appears in~{\ifnum\Includeappendix=0{the appendix}\else{Section \ref{sec:proof of lemma}}\fi}. We stress that this lemma holds regardless of Assumption~\ref{assumption:dominance}. Next, we leverage Lemma \ref{lemma:equivalence} to prove the main technical result of the paper.
\begin{theorem}
\label{thm:holy grail}
Let $\pi^*$ be a right-ordered, $\mP$-valid policy with $\sigr_{\pi^*}$ ordered in decreasing expected value. Under Assumption \ref{assumption:dominance}, for every state $s\in \mS=2^A$, it holds that $W(\pi^*,s)=W^*(s)$.
\end{theorem}
In particular, Theorem \ref{thm:holy grail} implies that $W(\SEGB,s_0)=W^*(s_0)$, since $\SEGB$ is right-ordered in decreasing expected value. The formal proof of Theorem \ref{thm:holy grail} is relegated to~{\ifnum\Includeappendix=0{the appendix}\else{Section \ref{sec:proof of thm}}\fi}. 

%% file: input/safe_exploration_lemma_proof.tex
\begin{proofof}{Lemma \ref{lemma:equivalence}}
We prove the lemma by a two-dimensional induction on the number of arms in $\above(s)$ and $\below(s)$. We prove four base cases in Section \ref{sec:base for lemma}:
\begin{itemize}
\item $\abs{\above(s)}=1$ and $\abs{\below(s)}\geq 2$ (Proposition \ref{prop:case of one strong}).
\item $\abs{\above(s)}\geq 2 $ and $\abs{\below(s)} = 1$ (Proposition \ref{prop:case of one}).
\item $\abs{\above(s)}\geq 2$ and $\abs{\below(s)} = 2$ (Proposition \ref{prop:case of two strong}).
\item $\abs{\above(s)}=2$ and $\abs{\below(s)}\geq 2$ (Proposition \ref{prop:case of two}).
\end{itemize}
While the first two cases are almost immediate, the other two are technical and require careful attention. Next, assume the statement holds for all states $s\in \mS$ such that $\abs{\above(s)}\leq K_1$, $\abs{\below(s)}\leq K_2$ and $\abs{\above(s)}+\abs{\below(s)}< K_1+K_2$.

Let $U\in\mS$ denote a state with $\abs{\above(U)}=K_1$ and $\abs{\below(U)}=K_2$. For abbreviation, let $\ug\defeq\above(U),\ul\defeq\below(U)$. Further, define $Q^*(U) = \sup_{\pi} Q(\pi,U)$, \footnote{This supremum is attained since there are only finitely many $\mP$-valid policies.} and for every $a_i \in \ug, a_j\in \ul$ let 
\[
Q^*_{i,j}(\ug,\ul)\defeq\bl p_{i,j}(a_j)Q^*(\ug,\ul\setminus \{a_j\}) +\bl p_{i,j}(a_i)Q^*(\ug\setminus\{a_i\},\ul).
\]
Next, let $(a_{i^*},a_{j^*})\in \argmax_{a_i\in \ug,a_j\in \ul}Q^*_{i,j}(\ug,\ul)$, and assume by contradiction that there exists a pair $(a_{\tilde i}, a_{\tilde j})$ such that 
\begin{equation}\label{eq:contradiction of lemma}
Q^*_{{i^*},{j^*}}(U) > Q^*_{{\tilde i}, {\tilde j}}(U).
\end{equation}
\paragraph{Step 1} Fix arbitrary $a_{i'}$ and $a_{j'}$ such that  $a_{i'} \in \ug$ and $a_{j'} \in \ul$. We will show that 
\begin{equation}\label{eq:step 1 goal}
Q^*_{{i'},{j^*}}(U)=Q^*_{{i'},{j'}}(U).
\end{equation}
We define the ordered policy $\pi$ such that $\sigr_\pi=(a_{i'},\dots)$, i.e., $\sigr_\pi$ first explores $a_{i'}$ and then the rest of the arms of $\ug$ in some arbitrary order; and, $\sigl_\pi=(a_{j^*},a_{j'},\dots)$. In addition, we define $\rho$ such that $\sigl_{\rho}=\sigl_\pi$, and $\sigr_\rho=(a_{j'},a_{j^*},\dots)$. Due to the inductive assumption, we have 
\begin{align}\label{eq:policies suffice}
Q^*_{{i'},{j^*}}(U) &=  \bl p_{{i'},{j^*}}(a_{j^*})Q^*(\ug,\ul\setminus \{a_{j^*}\}) +\bl p_{{i'},{j^*}}(a_{i'})Q^*(\ug\setminus \{a_{i'}\},\ul) \nonumber \\
&= \bl p_{{i'},{j^*}}(a_{j^*})Q(\pi,\ug,\ul\setminus \{a_{j^*}\})+\bl p_{{i'},{j^*}}(a_{i'})Q(\pi,\ug\setminus \{a_{i'}\},\ul)\\
&=Q(\pi,U).\nonumber
\end{align}
Similarly, $Q^*_{{i'},{j'}}(U)  = Q(\rho,U)$; hence, proving that $Q(\pi,U) =Q(\rho,U)$ entails Equality (\ref{eq:step 1 goal}). Next, let $\suff(\sigl_\pi)$ be the set of all non-empty suffices of $\sigl_\pi$. Being left-ordered suggests that on-path\footnote{These are terminal states that $\pi$ reaches to with positive probability.} terminal states with all arms of $\ul$ explored of $\pi$ are of the form $(Z,\emptyset)$, where $Z\in \suff(\sigl_\pi)$. Next, we factorize $Q(\pi,U)$ recursively as follows: We factorize $Q(\pi,U)$ into two terms, like in Equation \refeq{eq:policies suffice}. Following, for each term obtained, we ask whether the corresponding state excludes $\{a_{j^*},a_{j'}\}$. If the answer is yes, we stop factorizing it, and move to the other terms. We do this recursively, until we cannot factorize anymore, or we reached a terminal state. Using this factorizing process, we have \footnote{We stop factorizing if both $a_{j^*},a_{j'}$ were observed; thus, $Z$ will never be the empty set.} 
\begin{align*}
Q(\pi,U) &= \alpha \cdot Q(\pi,\emptyset,\ul)+\beta\cdot Q(\pi,\emptyset,\ul\setminus \{a_{j^*} \}) + \sum_{Z \in \suff(\sigl_\pi)} c^\pi_Z \cdot Q(\pi,Z,\ul\setminus \{a_{j^*},a_{j'}\}),
\end{align*}
for $\alpha=\Pr(\pathto{s}{(\emptyset,\ul)})$ and $\beta=\Pr(\pathto{s}{(\emptyset,\ul\setminus\{a_{j^*}\})})$ such that $\alpha+\beta+\sum_{Z\in \suff(\sigl_\pi)}c^\pi_Z =1$ and $\alpha,\beta,c^\pi_Z\in [0,1]$ for every $Z \in \suff(\sigl_\pi)$. In this representation, $\alpha$ is the probability of reaching the terminal $(\emptyset,\ul)$, while $\beta$ is the probability of reaching the terminal state $(\emptyset,\ul\setminus\{a_{j^*}\})$. For these two terminal states, we know that $Q^*(\emptyset,\ul)= Q^*(\emptyset,\ul\setminus \{a_{j^*} \})=0$; hence,
\begin{align}\label{eq:pi j^* to j}
Q(\pi,U) &= \sum_{Z \in \suff(\sigl_\pi)} c^\pi_Z \cdot Q(\pi,Z,\ul\setminus \{a_{j^*},a_{j'}\}).
\end{align}
Following the same factorization process for $\rho$, we get
\begin{align}\label{eq:rho j^* to j}
Q(\rho,U) &= \sum_{Z \in \suff(\sigl_\rho)} c^\rho_Z \cdot Q(\rho,Z,\ul\setminus \{a_{j^*},a_{j'}\}).
\end{align}
Next, we want to simplify the coefficients $\left(c^\pi_Z\right)_Z$. We remark that $c^\pi_Z$ is not simply the probability of reaching $(Z,\ul\setminus \{a_{j^*},a_{j'}\})$ from $s$, i.e., $\Pr(\pathto{s}{(Z,\ul\setminus \{a_{j^*},a_{j'}\})})$. To clarify. consider a strict suffix $Z$, $1\leq \abs{Z}< \abs{\above(A)}$, and the suffix $Z'=Z \cup\{a_l\}$ for the minimal element $a_l \in \ug \setminus Z$ according to $\sigl_\pi$,i.e., $a_l = \argmin_{a\in \ug \setminus Z}\sigl_\pi(a)$. In the factorization process that produced Equation (\ref{eq:pi j^* to j}), once we got the term $Q(\pi,Z',\ul\setminus \{a_{j^*},a_{j'}\})$, we stopped factorizing any further; thus, $c^\pi_Z$ does not include the probability of reaching a node associated with $(Z',\ul\setminus \{a_{j^*},a_{j'}\})$ and then following the left edge to $(Z,\ul\setminus \{a_{j^*},a_{j'}\})$. However, this probability is taken into account in $\Pr(\pathto{s}{(Z,\ul\setminus \{a_{j^*},a_{j'}\})})$. Rather, $c^\pi_Z$ is the probability of reaching any node $v$ in the tree induced by $\pi$ with the following property: $v$ represents the state $(Z,\ul\setminus \{a_{j^*},a_{j'}\})$, while $a_{j'}$ does not belong to the state represented by the parent of $v$.  In the tree interpretation, $v$ should also be a \textit{right child of its parent} (for instance, the left highlighted node in the tree in Figure \ref{fig:tree example}). The following Proposition \ref{prop:coef c} describes $\left(c^\pi_Z\right)_Z$ in terms of $Q$.
\begin{proposition}\label{prop:coef c}
For every $Z\in \suff(\sigl_\pi)$, let $a_{i(Z)} = \argmin_{a_i\in Z} \sigl_\pi(a_i)$. It holds that
\[
c^\pi_Z = Q(\pi,\ug\setminus Z \cup \{a_{i(Z)}\}, \{a_{j^*},a_{j'}\})-Q(\pi,\ug\setminus Z, \{a_{j^*},a_{j'}\}).
\]
\end{proposition}
The proof of Proposition \ref{prop:coef c} appears at the end of this proof. Notice that for every $Z$, $c^\pi_Z$ includes values of $Q$ with less arms than $U$ (besides, perhaps, the case where $\abs{\ug}=2$ and $\abs{Z}=1$ obtaining $Q(\pi,\ug, \{a_{j^*},a_{j'}\})$, but we cover this case in the bases cases); consequently, due to the inductive step
\begin{align}\label{c pi is rho}
c^\pi_Z = Q(\rho,\ug\setminus Z \cup \{a_{i(Z)}\}, \{a_{j^*},a_{j'}\})-Q(\rho,\ug\setminus Z, \{a_{j^*},a_{j'}\})=c^\rho_Z,
\end{align}
where the last equality follows from mirroring Proposition \ref{prop:coef c} for $(c^\rho_Z)_Z$. Ultimately,
{\thinmuskip=.2mu
\medmuskip=0mu plus .2mu minus .2mu
\thickmuskip=1mu plus 1mu
\begin{align*}
Q(\pi,U) &\stackrel{\textnormal{Eq. (\ref{eq:pi j^* to j})}}{=}\sum_{Z \in \suff(\sigl_\pi)}c^\pi_Z \cdot Q(\pi,Z,\ul\setminus \{a_{j^*},a_{j'}\})\stackrel{\textnormal{Eq. (\ref{c pi is rho})}}{=}\sum_{Z \in \suff(\sigl_\pi)}c^\rho_Z \cdot Q(\pi,Z,\ul\setminus \{a_{j^*},a_{j'}\}) \nonumber\\
&\stackrel{\textnormal{Ind. step}}{=}\sum_{Z \in \suff(\sigl_\pi)}c^\rho_Z \cdot Q(\rho,Z,\ul\setminus \{a_{j^*},a_{j'}\}) \stackrel{\sigl_{\rho}=\sigl_\pi}{=}\sum_{Z \in \suff(\sigl_\rho)}c^\rho_Z \cdot Q(\rho,Z,\ul\setminus \{a_{j^*},a_{j'}\}) \nonumber\\
&\stackrel{\textnormal{Eq. (\ref{eq:rho j^* to j})}}{=}Q(\rho,U) .
\end{align*}}
This completes Step 1.

\paragraph{Step 2}
Fix arbitrary $a_{i'}$ and $a_{j'}$ such that $a_{i'} \in \ug$ and $a_{j'} \in \ul$. We will show that 
\begin{equation}\label{eq:step 2 goal}
Q^*_{{i^*},{j'}}(U)=Q^*_{{i'},{j'}}(U).
\end{equation}
We follow the same technique as in the previous step. Let $\pi$ be an ordered policy such that $\sigl_\pi=(a_{i^*},a_{i'},\dots )$, i.e., $\sigl_\pi$ ranks $a_{i^*}$ first, $a_{i'}$ second and then follows some arbitrary order on the remaining arms, and $\sigr_\pi=(a_{j'},\dots)$. In addition, we define the ordered policy $\rho$ with $\sigl_\rho=(a_{i'},a_{i^*},\dots )$, where the dots refer to any arbitrary order on the remaining elements of $\ug$, and $\sigr_\rho =\sigr_\pi=(a_{j'},\dots)$. Using the inductive step and the same arguments as in Equation (\ref{eq:policies suffice}), it suffices to show that $Q(\pi,U)=Q(\rho,U)$. We factorize $Q(\pi,U)$ recursively such that \omer{Explain this factorization better?}
\begin{align}\label{eq: q pi with d}
Q(\pi,U)=Q(\pi,\{a_{i^*},a_{i'}\},\ul)+ \sum_{Z \in \suff(\sigr_\pi)} d^\pi_Z \cdot Q(\pi,\ug \setminus \{a_{i^*},a_{i'}\},Z),
\end{align}
and similarly
\begin{align}\label{eq: q rho with d}
Q(\rho,U)=Q(\rho,\{a_{i^*},a_{i'}\},\ul)+ \sum_{Z \in \suff(\sigr_\rho)} d^\rho_Z \cdot Q(\rho,\ug \setminus \{a_{i^*},a_{i'}\},Z).
\end{align}
Next, we claim that
\begin{proposition}\label{prop:coef d}
For every $Z\in \suff(\sigr_\pi)$, let $a_{j(Z)} = \argmin_{a_j\in Z}\sigr_\pi(a_j)$. It holds that
\[
d^\pi_Z = Q(\pi,\{a_{i^*},a_{i'}\},\ul \setminus Z )-Q(\pi,\{a_{i^*},a_{i'}\},\ul \setminus Z \cup \{a_{j(Z)}\}).
\]
\end{proposition}
The proof of Proposition \ref{prop:coef d} appears at the end of this proof. Notice that for every $Z$, $d^\pi_Z$ includes values of $Q$ with less arms than $U$ (besides, perhaps, the case where $\abs{\ul}=2$ and $\abs{Z}=1$ obtaining $Q(\pi,\{a_{i^*},a_{i'}\},\ul\})$, but we cover this case in the bases cases); consequently, due to the inductive step
\begin{align}\label{eq: d pi is rho}
d^\pi_Z = Q(\rho,\{a_{i^*},a_{i'}\},\ul \setminus Z )-Q(\rho,\{a_{i^*},a_{i'}\},\ul \setminus Z \cup \{a_{j(Z)}\})=d^\rho_Z,
\end{align}
where the last equality follows from mirroring Proposition \ref{prop:coef d} for $(d^\rho_Z)_Z$. Ultimately, by rearranging Equation (\ref{eq: q pi with d}) and invoking the inductive step, Equation (\ref{eq: d pi is rho}) and the fact that $\sigr_\rho =\sigr_\pi$, we get
\begin{align*}
Q(\pi,U)=Q(\rho,\{a_{i^*},a_{i'}\},\ul)+ \sum_{Z \in \suff(\sigr_\rho)} d^\rho_Z \cdot Q(\rho,\ug \setminus \{a_{i^*},a_{i'}\}),Z)\stackrel{\textnormal{Eq. (\ref{eq: q rho with d})}}{=}Q(\pi,U),
\end{align*}
implying Equation (\ref{eq:step 2 goal}) holds.
\paragraph{Step 3} We are ready to prove the lemma. Fix arbitrary $a_{\tilde i}$ and $a_{\tilde j}$ such that $a_{\tilde i} \in \ug$ and  $a_{\tilde j} \in \ul$. By the previous Step 1 and Step 2 we know that
\[
Q^*_{{i^*},{j^*}}(U)\stackrel{\textnormal{Step 1}}{=}Q^*_{{i^*},{\tilde j}}(U)\stackrel{\textnormal{Step 2}}{=}Q^*_{{\tilde i},{\tilde j}}(U),
\]
which contradicts Equation (\ref{eq:contradiction of lemma}); hence, the lemma holds.
\end{proofof}

%% file: input/safe_exploration_lemma_additions.tex
\begin{proofof}{Proposition \ref{prop:coef c}}
\begin{figure}
\centering
\includegraphics[scale=.95]{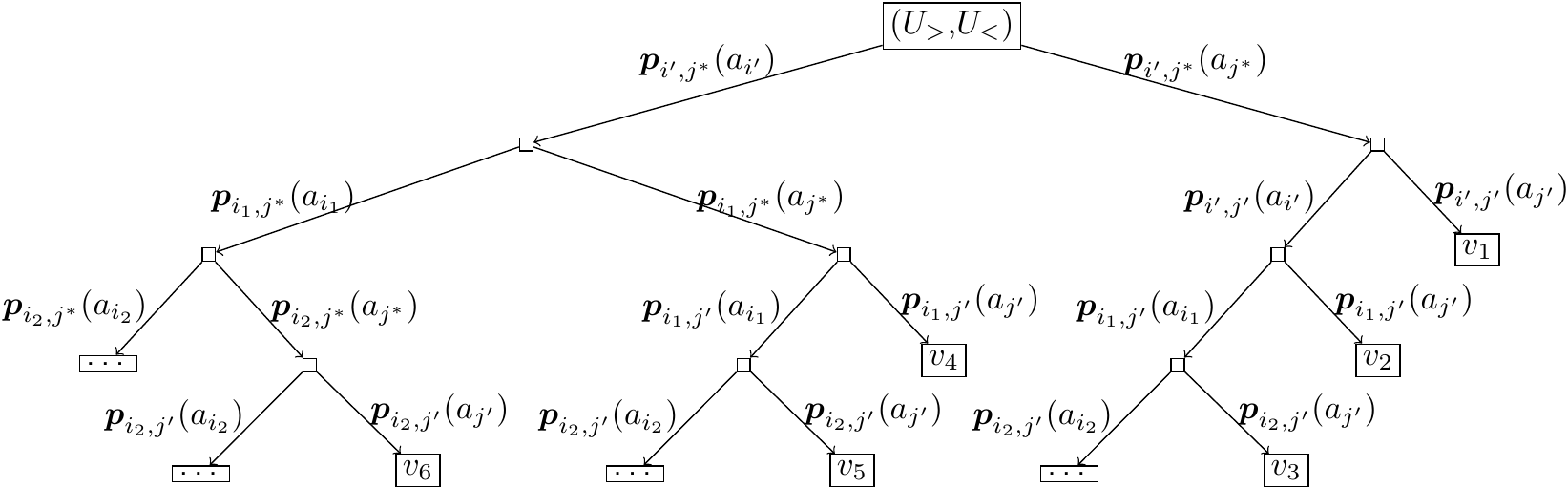}
\caption{Illustration for Proposition \ref{prop:coef c}. The tree depicts $T(\pi)$. Nodes $v_1$ to $v_6$ are nodes whose sub-trees were pruned in the construction of $T$. Let $Z=\ug\setminus \{a_{i'},a_{i_1}\}$ and $Z^c= \{a_{i'},a_{i_1}\}$. The minimal element of $Z$, denoted $a_{i(Z)}$ in the proof, is $a_{i_2}$. The corresponding $c^\pi_Z$ is the probability to reach one of $\{v_3,v_5,v_6\}$, namely, $c^\pi_Z = \Pr(\{v_3,v_5,v_6\})$. In the tree $T$, we ignore sub-trees of nodes $v$ labeled with ``$\dots $'' since these do not contribute to $c^\pi_Z$. Observe that the probability of reaching $v_i$, for $i\in\{1,\dots, 6\}$ is the same in $T(\pi)$ and $T$. Finally, notice that $Q(\pi,\ug\setminus Z \cup \{a_{i(Z)}\}, \{a_{j^*},a_{j'}\})=\Pr(\{v_1,v_2,\dots,v_6 \})$, and $Q(\pi,\ug\setminus Z, \{a_{j^*},a_{j'}\})=\Pr(\{v_1,v_2,v_4\})$. Combining, we get that $c^\pi_Z = Q(\pi,\ug\setminus Z \cup \{a_{i(Z)}\}, \{a_{j^*},a_{j'}\})-Q(\pi,\ug\setminus Z, \{a_{j^*},a_{j'}\})=\Pr(\{v_3,v_5,v_6\})$ as required.
\label{fig:tree illustation}}
\end{figure}
To prove this claim, we focus on the tree induced by $\pi$, $T(\pi)$. It is convenient to discuss a modified version $T(\pi)$ obtained by pruning, and this is feasible since even if prune nodes from $T(\pi)$ it still remains Markov chain. We illustrate the proof of this claim in Figure \ref{fig:tree illustation}.

We factorize $Q(\pi,U) $ recursively (see Equation (\ref{eq:pi j^* to j})) until we hit a node associated with a state that excludes $\{a_{j^*},a_{j'}\}$, or a leaf. This factorization can be illustrated as follows: We traverse $T(\pi)$, from right to left. Every node we visit, we ask whether that node includes $\{a_{j^*},a_{j'}\}$. If it does not, we prune its sub-tree (i.e., it becomes a leaf) while leaving it intact. Denote the obtained tree by $T$, and let $V(T)$ be its set of nodes. Observe that
\begin{observation}\label{obs:two types}
Every leaf $v$ in $V(T)$ satisfies exactly one property: 
\begin{enumerate}[leftmargin=0cm,itemindent=.5cm,labelwidth=\itemindent,labelsep=0cm,align=left]
\item[]\textit{Type 1:} $\above(state(v))=\emptyset$, or
\item[]\textit{Type 2:} $\below(state(v))=\ul\setminus \{j^*,j'\}$ with $\below(state(v)) \subset \below(state(parent(v)))$.
\end{enumerate}
\end{observation}
Leaves of type 1 are associated with terminal states of the MDP (see Subsection~\ref{subsec:aux GMDP}). Leaves of type 2 are those whose sub-trees were pruned during the traversal. Moreover, $\below(state(v)) \subset \below(state(parent(v)))$ holds in every such a leaf $v$, since otherwise we would have pruned its parent. Due to Observation \ref{obs:two types}, every node $v$ with $state(v)=(Z,\ul\setminus \{a_{j^*},a_{j'}\})$ is of type 2; therefore,
\begin{align}\label{eq:c with T}
c^\pi_Z = \sum_{\substack{v\in V(T):state(v)=\\(Z,\ul\setminus \{a_{j^*},a_{j'}\})}}\Pr\left(\pth{root(T)}{}{v}\right).
\end{align}

Next, fix an arbitrary non-empty $Z$, $Z \subseteq \ug$, and $\Psi$ be the set of all non empty suffixes of $\ug\setminus Z$. Consider $T$ and its root $root(T)$. Notice that $Q(\pi,\ug\setminus Z, \{a_{j^*},a_{j'}\})$ is the probability of reaching a (type 2) leaf $v$ such that $above(v)=\psi\cup Z$ for some $\psi \in \Psi$. This is true since $\pi$ is ordered, and every path from $root(T)$ to such a leaf $v$ does not include any action from $Z$; hence, we can compare the probability of reaching it to off-path behavior of $\pi$. Further, $Q(\pi,\ug\setminus Z \cup \{a_{i(Z)}\}, \{a_{j^*},a_{j'}\})$ is the probability of reaching a (type 2) leaf $v$ such that $above(v)=\psi\cup Z$ for some $\psi \in \Psi$ or $above(v)=Z$; hence,
\[
Q(\pi,\ug\setminus Z \cup \{a_{i(Z)}\}, \{a_{j^*},a_{j'}\})-Q(\pi,\ug\setminus Z, \{a_{j^*},a_{j'}\})
\]
is precisely the right-hand-side of Equation \refeq{eq:c with T}.
\end{proofof}

\begin{proofof}{Proposition \ref{prop:coef d}}
Fix $Z\in \suff(\sigr_\pi)$, and let $a_{j(Z)} = \argmin_{a_j\in Z}\sigr_\pi(a_j)$. Let $T(\pi)$ denote the tree induced by $\pi$. Observe that
\begin{observation}\label{obs: for d}
The coefficient $d^\pi_Z$ is the probability to get to a node $v$ in $T(\pi)$ such that
\begin{enumerate}
\item $state(v)=(\ug\setminus \{a_{i^*},a_{i'} \},Z )$, and 
\item $state(parent(v))=(\ug\setminus \{a_{i^*}\},Z )$. 
\end{enumerate}
\end{observation}
The first condition is immediate, due to the way we factorize $Q$ in Equation \refeq{eq: q pi with d}. To see why the second condition holds, notice that $state(parent(v))$ must be a strict superset of $state(v)$; hence, $state(parent(v))$ could be either $(\ug\setminus \{a_{i^*}\},Z )$ or $(\ug\setminus \{a_{i^*},a_{i'}\},Z\cup\{ a\} )$ for $a\in \ul\setminus Z$, but then it would contribute to $d^\pi_{Z \cup \{a\}}$, namely, to another summand in Equation \refeq{eq: q pi with d}. 

Denote by $V$ the set of all nodes that satisfy the conditions of Observation \ref{obs: for d}. Due to the way we constructed $\pi$, the paths from the root of $T(\pi)$ to any node in $V$ consist of actions that involve the arms $\{a_{i^*},a_{i'},a_{j(Z)}\}\cup (\ul \setminus Z)$ solely; hence, we can focus on the \text{off-path} tree whose root is $s_0'\defeq\{a_{i^*},a_{i'}\}\cup (\ul \setminus Z)\cup \{a_{j(Z)}\}$, and the actions are precisely as in the tree induced by $\pi$ (according to the order of $\pi$). Denote this new tree by $T'$, and let
\[
V' \defeq \left\{v\in nodes(T)\mid state(v)=\{a_{j(Z)}\}  \right\}.
\] 
Due to this construction, 
\begin{observation}\label{obs: for d two}
The coefficient $d^\pi_Z$ is the probability to get to a node that belongs to $V'$  in $T'$.
\end{observation}
The observation follows from the one-to-one correspondence between the nodes and path in $T(\pi)$ and their counterparts in the off-path tree $T'$. 

In $T'$,  $Q(\pi,s'_0)$ is the probability of starting at $s'_0$ and reaching the leaf with no arms (terminal state $\emptyset$), i.e., exploring $\ul \setminus Z$ \textit{and} $a_{j(Z)}$. In contrast, $Q(\pi,\{a_{i^*},a_{i'}\}\cup (\ul \setminus Z))$ is the probability of starting at $s'_0$ and reaching a node (internal or terminal) $v$ with $state(v)= \cap \left(\ul \setminus Z\right) =\emptyset$, namely, exploring $\ul \setminus Z$. Such a node $v$ leads to a leaf with probability 1; hence, paths from $v$ end in leaves associated with the states $\emptyset$ or  $\{a_{j(Z)}\}$ only. Consequently,
\[
Q(\pi,\{a_{i^*},a_{i'}\}\cup (\ul \setminus Z)) - Q(\pi,\{a_{i^*},a_{i'}\}\cup (\ul \setminus Z)\cup \{ a_{j(Z)}\})
\]
is the probability of starting at $s_0'$, and reaching a terminal node that belongs to $V'$.
\end{proofof}

%% file: input/safe_exploration_lemma_base.tex
\begin{proposition}\label{prop:case of one strong}
Let $\abs{\ug}=1$ and $\abs{\ul} \geq 2$. For any pair of policies $\pi,\rho$, it holds that $Q(\pi,\ug,\ul)=Q(\rho,\ug,\ul)$.
\end{proposition}
\begin{proofof}{Proposition \ref{prop:case of one strong}}
Let $\tilde \mu(a)\defeq \abs{\mu(a)}$, and denote $\ug=\{a_{i_1}\}$  and $\ul=\{a_{j_1},\dots a_{j_k}\}$ for $k=\abs{\ul}$. The probability of reaching the empty terminal state under any $\mP$-valid policy is
\begin{align*}
\prod_{l=1}^k\frac{\tilde \mu(a_{j_l})}{\tilde \mu(a_{j_l})+\tilde \mu(a_{i_1})},
\end{align*}
i.e., the probability of successfully exploring $\ul$. Due to multiplication associativity, the above expression is invariant of the way we order its elements. Finally, by definition of $Q$, this implies that $Q(\pi,\ug,\ul)=Q(\rho,\ug,\ul)$.
\end{proofof}

\begin{proposition}\label{prop:case of one}
Let $\abs{\ug}\geq 2$ and $\abs{\ul} = 1$. For any pair of policies $\pi,\rho$, it holds that $Q(\pi,\ug,\ul)=Q(\rho,\ug,\ul)$.
\end{proposition}
\begin{proofof}{Proposition \ref{prop:case of one}}
Let $\tilde \mu(a)\defeq \abs{\mu(a)}$, and denote $\ug=\{a_{i_1},\dots a_{i_k}\}$ for $k=\abs{\ug}$ and $\ul=\{a_{j_1}\}$. The probability of reaching the terminal state $(a_{j_1})$ under any $\mP$-valid policy is
\begin{align*}
\prod_{l=1}^k\frac{\tilde \mu(a_{i_l})}{\tilde \mu(a_{i_l})+\tilde \mu(a_{j_1})},
\end{align*}
i.e., the probability of failing to explore $a_{j_1}$. Due to multiplication associativity, the above expression is invariant of the way we order its elements. Finally, by definition of $Q$, this implies that $1-Q(\pi,\ug,\ul)=1-Q(\rho,\ug,\ul)$; hence, $Q(\pi,\ug,\ul)=Q(\rho,\ug,\ul)$
\end{proofof}

\begin{proposition}\label{prop:case of two strong}
Let $U$ be an arbitrary state, such that $\ug\defeq\above(U)=2$ and $\ul\defeq\below(U) \geq 2$. For any pair of $\mP$-valid policies $\pi$ and $\rho$, it holds that $Q(\pi,U)=Q(\rho,U)$.
\end{proposition}
\begin{proofof}{Proposition \ref{prop:case of two strong}}
We prove the claim by induction, with Proposition \ref{prop:case of one} serving as the base case. Assume the claim holds for $\abs{\ul}=k-1$. It is enough to show that if $\abs{\ul}=k$, for any $a_i\in \ug,a_j\in \ul$, $Q^*_{i,j}(U)=Q^*(U)$. Assume that $Q^*_{{i_1},{j_1}}(U)=Q^*(U)$, and fix any $a_{i'}\in \ug, a_{j'}\in \ul$. 
\paragraph{Remark} We do not use Assumption \ref{assumption:dominance} here.
\paragraph{Step 1} Assume that $i'=i_1$ and $j' \neq j_1$. W.l.o.g. $j'=j_2$. We construct two policies, $\pi$ that ordered $\ul$ as $\sigr_\pi=(a_{j_1},a_{j_2},\dots, a_{j_{k}})$, and $\rho$ that orders $\ul$ as $\sigr_\rho=(a_{j_2},a_{j_1},\dots,a_{j_{k}})$. Both policies order $\ug$ according to $\sigl_\pi=\sigl_\rho=(a_{i_1},a_{i_2})$. Due to the inductive step and our assumption that  $Q^*_{i_1,j_1}(U)=Q^*(U)$, we have that $Q(\pi,U)=Q^*(U)$, and
\begin{align*}
&Q(\pi,\ug,\ul) = \underbrace{\prod_{l=1}^{k} \bl p_{{i_1},{j_l}}(a_{j_l})}_{\lambda(\pi)}
+ \underbrace{\sum_{f=1}^{k} \left(\prod_{l=1}^{f-1} \bl p_{{i_1},{j_l}}(a_{j_l})\right) \bl p_{{i_1},{j_f}}(a_{i_1}) \left(\prod_{l=f}^{k} \bl p_{{i_2},{j_l}}(a_{j_l})\right)}_{\delta(\pi)}.
\end{align*}
Notice that $\lambda(\pi)=\lambda(\rho)$. In addition, we have that
{\thinmuskip=0mu
\medmuskip=0mu plus 0mu minus 0mu
\thickmuskip=0mu plus 0mu
\begin{align}\label{eq:ind 2 step 1}
\delta(\pi)&=\bl p_{{i_1},{j_1}}(a_{i_1})\prod_{l=1}^{k} \bl p_{{i_2},{j_l}}(a_{j_l})\nonumber\\
&\qquad +\bl p_{{i_1},{j_1}}(a_{j_1})\sum_{f=2}^{k} \left(\prod_{l=2}^{f-1} \bl p_{{i_1},{j_l}}(a_{j_l})\right) \bl p_{{i_1},{j_f}}(a_{i_1}) \left(\prod_{l=f}^{k} \bl p_{{i_2},{j_l}}(a_{j_l})\right)\nonumber\\
& =\bl p_{{i_1},{j_1}}(a_{i_1})\bl p_{{i_2},{j_1}}(a_{j_1})\bl p_{{i_2},{j_2}}(a_{j_2})\prod_{l=3}^{k} \bl p_{{i_2},{j_l}}(a_{j_l})\nonumber\\
&\qquad+ \bl p_{{i_1},{j_1}}(a_{j_1})\Bigg[\bl p_{{i_1},{j_2}}(a_{i_1})\bl p_{{i_2},{j_2}}(a_{j_2})\prod_{l=3}^{k} \bl p_{{i_2},{j_l}}(a_{j_l})\nonumber\\
&\qquad \qquad+\bl p_{{i_1},{j_2}}(a_{j_2})  \sum_{f=3}^{k} \left(\prod_{l=3}^{f-1} \bl p_{{i_1},{j_l}}(a_{j_l})\right) \bl p_{{i_1},{j_f}}(a_{i_1}) \left(\prod_{l=f}^{k} \bl p_{{i_2},{j_l}}(a_{j_l})\right) \Bigg]\nonumber\\
&=\bl p_{{i_2},{j_2}}(a_{j_2})\left( \bl p_{{i_1},{j_1}}(a_{i_1})\bl p_{{i_2},{j_1}}(a_{j_1})+\bl p_{{i_1}
,{j_1}}(a_{j_1})\bl p_{{i_1},{j_2}}(a_{i_1})\right)\prod_{l=3}^{k} \bl p_{{i_2},{j_l}}(a_{j_l})\nonumber\\
&\qquad+ \bl p_{{i_1},{j_1}}(a_{j_1})\bl p_{{i_1},{j_2}}(a_{j_2})  \sum_{f=3}^{k} \left(\prod_{l=3}^{f-1} \bl p_{{i_1},{j_l}}(a_{j_l})\right) \bl p_{{i_1},{j_f}}(a_{i_1}) \left(\prod_{l=f}^{k} \bl p_{{i_2},{j_l}}(a_{j_l})\right).
\end{align}}%
We show that
\begin{claim}\label{claim:triplets}
It holds that
\begin{align*}
&\bl p_{{i_2},{j_2}}(a_{j_2})\left( \bl p_{{i_1},{j_1}}(a_{i_1})\bl p_{{i_2},{j_1}}(a_{j_1})+\bl p_{{i_1},{j_1}}(a_{j_1})\bl p_{{i_1},{j_2}}(a_{i_1})\right)\\
&=\bl p_{{i_2},{j_1}}(a_{j_1})\left(\bl p_{{i_1},{j_2}}(a_{i_1})\bl p_{{i_2},{j_2}}(a_{j_2})+\bl p_{{i_1},{j_2}}(a_{j_2})\bl p_{{i_1},{j_1}}(a_{i_1})  \right).
\end{align*}
\end{claim}
Now, set $\sigma:\mathbb N \rightarrow \mathbb N$ such that $\sigma(1)=2, \sigma(2)=1$, and $\sigma(i)=i$ for $i\geq 3$; hence, using Claim \ref{claim:triplets},
\begin{align}
\textnormal{Eq. (\ref{eq:ind 2 step 1})}
& =\bl p_{{i_1},{j_2}}(a_{i_1})\bl p_{{i_2},{j_2}}(a_{j_2})\bl p_{{i_2},{j_1}}(a_{j_1})\prod_{l=3}^{k} \bl p_{{i_2},{j_l}}(a_{j_l})\nonumber\\
&\qquad+ \bl p_{{i_1},{j_2}}(a_{j_2}) \Bigg[\bl p_{{i_1},{j_1}}(a_{i_1})\bl p_{{i_2},{j_1}}(a_{j_1}) \prod_{l=3}^{k} \bl p_{{i_2},{j_l}}(a_{j_l})\nonumber\\
&\qquad \qquad + \bl p_{{i_1},{j_1}}(a_{j_1})\sum_{f=3}^{k} \left(\prod_{l=3}^{f-1} \bl p_{{i_1},{j_l}}(a_{j_l})\right) \bl p_{{i_1},{j_f}}(a_{i_1}) \left(\prod_{l=f}^{k} \bl p_{{i_2},{j_l}}(a_{j_l})\right) \Bigg]\nonumber\\
&= \sum_{f=1}^{k} \left(\prod_{l=1}^{f-1} \bl p_{{i_1},{j_{\sigma(l)}}}(a_{j_{\sigma(l)}})\right) \bl p_{{i_1},{j_{\sigma (f)}}}(a_{i_1}) \left(\prod_{l=f}^{k} \bl p_{{i_2},{j_{\sigma(l)}}}(a_{j_{\sigma(l)}})\right)\nonumber\\
&=\delta(\rho)
\end{align}
\paragraph{Step 2} Assume that $i'=i_2\neq i_1$ and $j'=j_1$. We construct two policies, $\pi$ that orders $\ug$ as $\sigl_\pi=(a_{i_1},a_{i_2})$, and $\rho$ that orders $\ug$ as  $\sigl_\rho=(a_{i_2},a_{i_1})$. Both policies order $\ul$ by $\sigr_\pi=\sigr_\rho=(a_{j_1},a_{j_2}, \dots, a_{j_k})$. In addition, we introduce a third policy, $\tilde \rho$, that has the same order as $\rho$ on $\ug$, namely $\sigl_{\tilde \rho}=\sigl_\rho=(a_{i_2},a_{i_1})$, and orders $\ul$ by $\sigr_{\tilde \rho}=(a_{j_k},a_{j_2},\dots ,a_{j_{k-1}},a_{j_1})$.
It holds that
{\thinmuskip=0mu
\medmuskip=0mu plus 0mu minus 0mu
\thickmuskip=0mu plus 0mu
\begin{align}\label{eq: multiple transitions}
&Q(\pi,\ug,\ul)=\prod_{l=1}^k \bl p_{{i_1},{j_l}}(a_{j_l})+\sum_{f=1}^{k} \left(\prod_{l=1}^{f-1} \bl p_{{i_1},{j_l}}(a_{j_l})\right) \bl p_{{i_1},{j_f}}(a_{i_1}) \left(\prod_{l=f}^{k} \bl p_{{i_2},{j_l}}(a_{j_l})\right)\nonumber\\
&=\bl p_{{i_1},{j_1}}(a_{j_1})\bl p_{{i_1},{j_k}}(a_{j_k})\overbrace{\prod_{l=2}^{k-1} \bl p_{{i_1},{j_l}}(a_{j_l})}^{I_1}+\bl p_{{i_2},{j_k}}(a_{j_k})\Bigg[\bl p_{{i_1},{j_1}}(a_{j_1})\bl p_{{i_1},{j_k}}(a_{i_1})\underbrace{\prod_{l=2}^{k-1} \bl p_{{i_1},{j_l}}(a_{j_l})}_{I_1}\nonumber\\
&\quad+\underbrace{\sum_{f=1}^{k-1} \left(\prod_{l=1}^{f-1} \bl p_{{i_1},{j_l}}(a_{j_l})\right) \bl p_{{i_1},{j_f}}(a_{i_1}) \left(\prod_{l=f}^{k-1} \bl p_{{i_2},{j_l}}(a_{j_l})\right)}_{I_2}\Bigg]\nonumber\\
&= \bl p_{{i_1},{j_1}}(a_{j_1})\bl p_{{i_1},{j_k}}(a_{j_k})\bl p_{{i_2},{j_k}}(a_{j_k})I_1 +\bl p_{{i_1},{j_1}}(a_{j_1})\bl p_{{i_1},{j_k}}(a_{j_k})\bl p_{{i_2},{j_k}}(a_{i_2})I_1\nonumber\\
&\qquad+\bl p_{{i_2},{j_k}}(a_{j_k})\left[\bl p_{{i_1},{j_1}}(a_{j_1})\bl p_{{i_1},{j_k}}(a_{i_1})I_1+I_2\right]\nonumber\\
&=\bl p_{{i_1},{j_1}}(a_{j_1})\bl p_{{i_1},{j_k}}(a_{j_k})\bl p_{{i_2},{j_k}}(a_{i_2})I_1\nonumber\\
&\qquad +\bl p_{{i_2},{j_k}}(a_{j_k})\left[\bl p_{{i_1},{j_1}}(a_{j_1})\bl p_{{i_1},{j_k}}(a_{j_k})I_1+\bl p_{{i_1},{j_1}}(a_{j_1})\bl p_{{i_1},{j_k}}(a_{i_1})I_1+I_2\right]\nonumber\\
&=\bl p_{{i_1},{j_1}}(a_{j_1})\bl p_{{i_1},{j_k}}(a_{j_k})\bl p_{{i_2},{j_k}}(a_{i_2})I_1+\bl p_{{i_2},{j_k}}(a_{j_k})\left[\bl p_{{i_1},{j_1}}(a_{j_1})I_1+I_2\right].
\end{align}}%
Next, observe that
{\thinmuskip=0mu
\medmuskip=0mu plus 0mu minus 0mu
\thickmuskip=0mu plus 0mu
\begin{equation}\label{eq:ind step with k-1}
\left[\bl p_{{i_1},{j_1}}(a_{j_1})I_1+I_2\right]=\prod_{l=1}^{k-1} \bl p_{{i_1},{j_l}}(a_{j_l})+\sum_{f=1}^{k-1} \left(\prod_{l=1}^{f-1} \bl p_{{i_1},{j_l}}(a_{j_l})\right) \bl p_{{i_1},{j_f}}(a_{i_1}) \left(\prod_{l=f}^{k-1} \bl p_{{i_2},{j_l}}(a_{j_l})\right).
\end{equation}}%
Notice that the latter is precisely $Q(\pi,\left\{a_{i_1},a_{i_2},a_{j_1},\dots a_{j_{k-1}}\right\})$; thus, the inductive step implies that it is order invariant. Let $\sigma:\mathbb N \rightarrow \mathbb N$ such that $\sigma(1)=k,\sigma(k)=1$, and $\sigma(i)=i$ for $1<i<k$. Since $Q(\pi,\left\{a_{i_1},a_{i_2},a_{j_1},\dots a_{j_{k-1}}\right\})=Q(\tilde \rho,\left\{a_{i_1},a_{i_2},a_{j_1},\dots a_{j_{k-1}}\right\})$, we conclude that the expression in Equation (\ref{eq:ind step with k-1}) equals 
\[
\prod_{l=2}^{k} \bl p_{{i_2},{j_{\sigma(l)}}}(a_{j_{\sigma(l)}})+\sum_{f=2}^{k} \left(\prod_{l=2}^{f-1} \bl p_{{i_2},{j_{\sigma(l)}}}(a_{j_{\sigma(l)}})\right) \bl p_{{i_2},{j_{\sigma(f)}}}(a_{i_2}) \left(\prod_{l=f}^{k} \bl p_{{i_1},{j_{\sigma(l)}}}(a_{j_{\sigma(l)}})\right).
\]
Combining this with Equation (\ref{eq: multiple transitions}),
{\thinmuskip=0mu
\medmuskip=0mu plus 0mu minus 0mu
\thickmuskip=0mu plus 0mu
\begin{align}
&\textnormal{Eq. }\refeq{eq: multiple transitions}=\bl p_{{i_1},{j_1}}(a_{j_1})\bl p_{{i_1},{j_k}}(a_{j_k})\bl p_{{i_2},{j_k}}(a_{i_2})\prod_{l=2}^{k-1} \bl p_{{i_1},{j_{\sigma(l)}}}(a_{j_{\sigma(l)}})\nonumber\\
&\quad +\bl p_{{i_2},{j_k}}(a_{j_k})\Bigg[
\prod_{l=2}^{k} \bl p_{{i_2},{j_{\sigma(l)}}}(a_{j_{\sigma(l)}})+\sum_{f=2}^{k} \left(\prod_{l=2}^{f-1} \bl p_{{i_2},{j_{\sigma(l)}}}(a_{j_{\sigma(l)}})\right) \bl p_{{i_2},{j_{\sigma(f)}}}(a_{i_2}) \left(\prod_{l=f}^{k} \bl p_{{i_1},{j_{\sigma(l)}}}(a_{j_{\sigma(l)}})\right)
\Bigg]\nonumber\\
&=\bl p_{{i_2},{j_k}}(a_{i_2})\prod_{l=1}^{k} \bl p_{{i_1},{j_{\sigma(l)}}}(a_{j_{\sigma(l)}})+\bl p_{{i_2},{j_k}}(a_{j_k})
\prod_{l=2}^{k} \bl p_{{i_2},{j_{\sigma(l)}}}(a_{j_{\sigma(l)}})\nonumber\\
&\quad +\bl p_{{i_2},{j_k}}(a_{j_k})\left[
\sum_{f=2}^{k} \left(\prod_{l=2}^{f-1} \bl p_{{i_2},{j_{\sigma(l)}}}(a_{j_{\sigma(l)}})\right) \bl p_{{i_2},{j_{\sigma(f)}}}(a_{i_2}) \left(\prod_{l=f}^{k} \bl p_{{i_1},{j_{\sigma(l)}}}(a_{j_{\sigma(l)}})\right)
\right]\nonumber\\
&=\prod_{l=1}^{k} \bl p_{{i_2},{j_{\sigma(l)}}}(a_{j_{\sigma(l)}})+
\sum_{f=1}^{k} \left(\prod_{l=1}^{f-1} \bl p_{{i_2},{j_{\sigma(l)}}}(a_{j_{\sigma(l)}})\right) \bl p_{{i_2},{j_{\sigma(f)}}}(a_{i_2}) \left(\prod_{l=f}^{k} \bl p_{{i_1},{j_{\sigma(l)}}}(a_{j_{\sigma(l)}})\right)\nonumber\\
&=Q(\tilde \rho ,U),
\end{align}}%
where the last equality follows from the definition of $\tilde \rho$ (orders the arms precisely so). Finally, $Q(\tilde \rho, U)=Q(\rho ,U)$ follows from the previous Step 1.
\paragraph{Step 3}
The two previous steps imply that for any $a_{i'}\in \ug, a_{j'}\in \ul$, it holds that
\[
Q_{i',j'}(U)=Q_{i',j_1}(U)=Q_{i_1,j_1}(U).
\]
This completes the proof of Proposition \ref{prop:case of two strong}.
\end{proofof}
\begin{proofof}{Claim \ref{claim:triplets}}
To ease readability, let $\tmu{i}\defeq\abs{\mu(a_i)}$ for every $a_i\in A$. It holds that
\begin{align*}
&\bl p_{{i_2},{j_2}}({j_2})\left( \bl p_{{i_1},{j_1}}({i_1})\bl p_{{i_2},{j_1}}({j_1})+\bl p_{{i_1},{j_1}}({j_1})\bl p_{{i_1},{j_2}}({i_1})\right)\\
&=\frac{\tmu{i_2}}{\tmu{i_2}+\tmu{j_2}}\left( \frac{\tmu{j_1}}{\tmu{i_1}+\tmu{j_1}}\frac{\tmu{i_2}}{\tmu{i_2}+\tmu{j_1}}+\frac{\tmu{i_1}}{\tmu{i_1}+\tmu{j_1}}\frac{\tmu{j_2}}{\tmu{i_1}+\tmu{j_2}}\right) \\
&= \frac{\tmu{i_2}\tmu{j_1}\tmu{i_2}(\tmu{i_1}+\tmu{j_2})+\tmu{i_2}\tmu{i_1}\tmu{j_2}(\tmu{i_2}+\tmu{j_1})}{(\tmu{i_1}+\tmu{j_1})(\tmu{i_1}+\tmu{j_2})(\tmu{i_2}+\tmu{j_1})(\tmu{i_2}+\tmu{j_2})}\\
&= \frac{\overbrace{\tmu{i_2}\tmu{j_1}\tmu{i_2}\tmu{i_1}}^{I}+\overbrace{\tmu{i_2}\tmu{j_1}\tmu{i_2}\tmu{j_2}}^{II}+\overbrace{\tmu{i_2}\tmu{i_1}\tmu{j_2}\tmu{i_2}}^{III}+\overbrace{\tmu{i_2}\tmu{i_1}\tmu{j_2}\tmu{j_1}}^{IV}}{(\tmu{i_1}+\tmu{j_1})(\tmu{i_1}+\tmu{j_2})(\tmu{i_2}+\tmu{j_1})(\tmu{i_2}+\tmu{j_2})}\\
&= \frac{\overbrace{\tmu{i_2}\tmu{j_2}\tmu{i_2}\tmu{i_1}}^{III}+\overbrace{\tmu{i_2}\tmu{j_2}\tmu{i_2}\tmu{j_1}}^{II}+\overbrace{\tmu{i_2}\tmu{i_1}\tmu{j_1}\tmu{i_2}}^{I}+\overbrace{\tmu{i_2}\tmu{i_1}\tmu{j_1}\tmu{j_2}}^{IV}}{(\tmu{i_1}+\tmu{j_1})(\tmu{i_1}+\tmu{j_2})(\tmu{i_2}+\tmu{j_1})(\tmu{i_2}+\tmu{j_2})}\\
&= \frac{\tmu{i_2}\tmu{j_2}\tmu{i_2}(\tmu{i_1}+\tmu{j_1})+\tmu{i_2}\tmu{i_1}\tmu{j_1}(\tmu{i_2}+\tmu{j_2})}{(\tmu{i_1}+\tmu{j_1})(\tmu{i_1}+\tmu{j_2})(\tmu{i_2}+\tmu{j_1})(\tmu{i_2}+\tmu{j_2})}\\
&=\bl p_{{i_2},{j_1}}({j_1})\left(\bl p_{{i_1},{j_2}}({i_1})\bl p_{{i_2},{j_2}}({j_2})+\bl p_{{i_1},{j_2}}({j_2})\bl p_{{i_1},{j_1}}({i_1})  \right)
\end{align*}
\end{proofof}

\begin{proposition}\label{prop:case of two}
Let $U$ be an arbitrary state, such that $\ug\defeq\above(U)\geq 2$ and $\ul\defeq\below(U) = 2$. For any pair of $\mP$-valid policies $\pi$ and $\rho$, it holds that $Q(\pi,U)=Q(\rho,U)$.
\end{proposition}
\begin{proofof}{Proposition \ref{prop:case of two}}
The proof of this proposition goes along the lines of the proof of Proposition \ref{prop:case of two strong}, but we provide the details here for completeness. For simplicity, we let $\overline Q = 1-Q$, and prove that for any two policies $\pi, \rho$ it holds that $\overline Q(\pi,U)=\overline Q(\rho,U)$.

We prove the claim by induction, with Proposition \ref{prop:case of one strong} serving as the base case. Assume the claim holds for $\abs{\ug}=k-1$. It is enough to show that if $\abs{\ug}=k$, for any $a_i\in \ug,a_j\in \ul$, $Q^*_{i,j}(U)=Q^*(U)$. Assume that $Q^*_{{i_1},{j_1}}(U)=Q^*(U)$, and define a policy $\pi$ such that $\sigr_\pi=(a_{j_1},a_{j_2})$ and $\sigl_\pi=(a_{i_1},a_{i_2},\dots a_{i_k})$. Next, fix any $a_{i'}\in \ug, a_{j'}\in \ul$. 
\paragraph{Remark} We do not use Assumption \ref{assumption:dominance} here.
\paragraph{Step 1} Assume that $i'\neq i_1$ and $j' = j_1$. W.l.o.g. $i'=i_2$. We construct the ordered policy $\rho$ that orders $\ul$ by $\sigr_\rho=\sigr_\pi=(a_{j_1},a_{j_2})$ and $\ug$ by $\sigl_\rho=(a_{i_2},a_{i_1},\dots, a_{i_k})$. Due to the inductive step and our assumption that  $Q^*_{i_1,j_1}(U)=Q^*(U)$, we have that $\overline Q(\pi,U)=\overline Q^*(U)$. For brevity, we introduce the following notations. For $r\in \{1,2,3\}$, let
{\thinmuskip=.2mu
\medmuskip=0mu plus .2mu minus .2mu
\thickmuskip=1mu plus 1mu
\[
\lambda^r_{j_1}=\prod_{l=r}^{k} \bl p_{{i_l},{j_1}}(a_{i_l}),\lambda^r_{j_2}=\prod_{l=r}^{k} \bl p_{{i_l},{j_2}}(a_{i_l}), \delta^r= \sum_{f=r}^{k} \left(\prod_{l=r}^{f-1} \bl p_{{i_l},{j_1}}(a_{i_l})\right) \bl p_{{i_f},{j_1}}(a_{j_1}) \left(\prod_{l=f}^{k} \bl p_{{i_l},{j_2}}(a_{i_l})\right).
\]
}%
Observe that
\begin{align}\label{eq:from pi to ghjghj}
\overline Q(\pi,U)&=\lambda^1_{j_1}+\delta^1=\lambda^1_{j_1}+\bl p_{i_1,j_1}(a_{j_1})\lambda^1_{j_2}+\bl p_{i_1,j_1}(a_{i_1})\delta^2\nonumber\\
&=\lambda^1_{j_1}+\bl p_{i_1,j_1}(a_{j_1})\bl p_{i_1,j_2}(a_{i_1})\bl p_{i_2,j_2}(a_{i_2})\lambda^3_{j_2}\nonumber\\
&\qquad \qquad+\bl p_{i_1,j_1}(a_{i_1})\left[
\bl p_{i_2,j_1}(a_{j_1})\bl p_{i_2,j_2}(a_{i_2})\lambda^3_{j_2}+\bl p_{i_2,j_1}(a_{i_2})\delta^3\right].
\end{align}
Next, we show that 
\begin{claim}\label{claim:triplets additional}
It holds that
\begin{align*}
&\bl p_{{i_2},{j_2}}(a_{i_2})\left( \bl p_{{i_1},{j_1}}(a_{j_1})\bl p_{{i_1},{j_2}}(a_{i_1})+\bl p_{{i_1},{j_1}}(a_{i_1})\bl p_{{i_2},{j_1}}(a_{j_1})\right)\\
&=\bl p_{{i_1},{j_2}}(a_{i_1})\left(\bl p_{{i_2},{j_1}}(a_{j_1})\bl p_{{i_2},{j_2}}(a_{i_2})+\bl p_{{i_2},{j_1}}(a_{i_2})\bl p_{{i_1},{j_1}}(a_{j_1})  \right).
\end{align*}
\end{claim}
Combining Equation \refeq{eq:from pi to ghjghj} and Claim \ref{claim:triplets additional}, we get
\begin{align}\label{eq: flipping}
\textnormal{Eq. } \refeq{eq:from pi to ghjghj}
&=\lambda^1_{j_1}+\bl p_{i_2,j_1}(a_{j_1})\bl p_{i_2,j_2}(a_{i_2})\bl p_{i_1,j_2}(a_{i_1})\lambda^3_{j_2}\nonumber\\
&\qquad \qquad+\bl p_{i_2,j_1}(a_{i_2})\left[
\bl p_{i_1,j_1}(a_{j_1})\bl p_{i_1,j_2}(a_{i_1})\lambda^3_{j_2}+\bl p_{i_1,j_1}(a_{i_1})\delta^3\right]\nonumber \\
&=\lambda^1_{j_1}+\bl p_{i_2,j_1}(a_{j_1})\lambda^1_{j_2}+\bl p_{i_2,j_1}(a_{i_2})\left[
\bl p_{i_1,j_1}(a_{j_1})\bl p_{i_1,j_2}(a_{i_1})\lambda^3_{j_2}+\bl p_{i_1,j_1}(a_{i_1})\delta^3\right]
\end{align}
Now, set $\sigma:\mathbb N \rightarrow \mathbb N$ such that $\sigma(1)=2, \sigma(2)=1$, and $\sigma(i)=i$ for $i\geq 3$; hence,
{\thinmuskip=.2mu
\medmuskip=0mu plus .2mu minus .2mu
\thickmuskip=1mu plus 1mu
\begin{align*}
\textnormal{Eq. }\refeq{eq: flipping}&=\lambda^1_{j_1}+\bl p_{i_2,j_1}(a_{j_1})\lambda^1_{j_2}\\
&\qquad \qquad +\bl p_{i_2,j_1}(a_{i_2})\left[ 
\sum_{f=2}^{k} \left(\prod_{l=2}^{f-1} \bl p_{{i_{\sigma(l)}},{j_1}}(a_{i_{\sigma(l)}})\right) \bl p_{{i_{\sigma (f)}},{j_1}}(a_{j_1}) \left(\prod_{l=f}^{k} \bl p_{{i_{\sigma(l)}},{j_2}}(a_{i_{\sigma(l)}})\right)
 \right] \\
&= \lambda^1_{j_1}+ \sum_{f=1}^{k} \left(\prod_{l=1}^{f-1} \bl p_{{i_{\sigma(l)}},{j_1}}(a_{i_{\sigma(l)}})\right) \bl p_{{i_{\sigma (f)}},{j_1}}(a_{j_1}) \left(\prod_{l=f}^{k} \bl p_{{i_{\sigma(l)}},{j_2}}(a_{i_{\sigma(l)}})\right)\\
&= \overline Q(\rho,U).
\end{align*}
}
This concludes Step 1.
\paragraph{Step 2} Assume that $i'=i_1$ and $j' = j_2 \neq j_1$. We construct two ordered policies, $\rho$ and $\tilde \rho$ such that $\sigr_\rho=\sigr_{\tilde \rho}=(a_{j_2},a_{j_1})$ and $\sigl_\rho=(a_{i_1},a_{i_2},\dots, a_{i_k})$, $\sigl_{\tilde \rho}=(a_{i_k},a_{i_2},\dots, a_{i_{k-1}},a_{i_1})$. The previous Step 1 implies that $\overline Q(\rho,U)=\overline Q(\tilde \rho, U)$; thus, it suffices to show that $\overline Q(\pi,U)=\overline Q(\tilde \rho, U)$. Notice that
{\thinmuskip=.2mu
\medmuskip=0mu plus .2mu minus .2mu
\thickmuskip=1mu plus 1mu
\begin{align}\label{jnknkfdas}
\overline Q(\pi,U)&=\bl p_{i_k,j_1}(a_{i_k})\overbrace{\prod_{l=1}^{k-1} \bl p_{{i_l},{j_1}}(a_{i_l})}^{I_1}\nonumber \\
&+\bl p_{i_k,j_2}(a_{i_k})\left[\bl p_{i_k,j_1}(a_{j_1})\underbrace{\prod_{l=1}^{k-1} \bl p_{{i_l},{j_1}}(a_{i_l})}_{I_1}+\underbrace{\sum_{f=1}^{k-1} \left(\prod_{l=1}^{f-1} \bl p_{{i_l},{j_1}}(a_{i_l})\right) \bl p_{{i_f},{j_1}}(a_{j_1}) \left(\prod_{l=f}^{k-1} \bl p_{{i_l},{j_2}}(a_{i_l})\right)}_{I_2}  \right]
\end{align}}%
Rearranging,
\begin{align}\label{eq:kijnjvfvs}
\textnormal{Eq. }\refeq{jnknkfdas} &=\bl p_{i_k,j_1}(a_{i_k})I_1+\bl p_{i_k,j_2}(a_{i_k})\left[\bl p_{i_k,j_1}(a_{j_1})I_1+I_2  \right]. \nonumber \\
&= \bl p_{i_k,j_1}(a_{i_k})\left( \bl p_{i_k,j_2}(a_{i_k})+\bl p_{i_k,j_2}(a_{j_2}) \right) I_1+\bl p_{i_k,j_2}(a_{i_k})\left[\bl p_{i_k,j_1}(a_{j_1})I_1+I_2  \right]. \nonumber \\
&= \bl p_{i_k,j_1}(a_{i_k})\bl p_{i_k,j_2}(a_{j_2})I_1+\bl p_{i_k,j_2}(a_{i_k})\left[I_1+I_2  \right]. 
\end{align}
Recall that
\begin{align} \label{porenkjnf}
I_1+I_2=\prod_{l=1}^{k-1} \bl p_{{i_l},{j_1}}(a_{i_l})+\sum_{f=1}^{k-1} \left(\prod_{l=1}^{f-1} \bl p_{{i_l},{j_1}}(a_{i_l})\right) \bl p_{{i_f},{j_1}}(a_{j_1}) \left(\prod_{l=f}^{k-1} \bl p_{{i_l},{j_2}}(a_{i_l})\right),
\end{align}
which is precisely $\overline Q(\pi,\left\{a_{i_1},\dots, a_{i_{k-1}},a_{j_1},a_{j_2}\right\})$; thus, the inductive step implies that it is order invariant. Let $\sigma:\mathbb N \rightarrow \mathbb N$ such that $\sigma(1)=k,\sigma(k)=1$, and $\sigma(i)=i$ for $1<i<k$. Since $\overline Q(\pi,\left\{a_{i_1},\dots, a_{i_{k-1}},a_{j_1},a_{j_2}\right\})=\overline Q(\tilde \rho,\left\{a_{i_1},\dots, a_{i_{k-1}},a_{j_1},a_{j_2}\right\}))$, we conclude that the expression in Equation (\ref{porenkjnf}) equals 
\[
\prod_{l=2}^{k} \bl p_{{\sigma (i_l)},{j_2}}(a_{\sigma (i_l)})+\sum_{f=2}^{k} \left(\prod_{l=2}^{f-1} \bl p_{{\sigma (i_l)},{j_2}}(a_{\sigma (i_l)})\right) \bl p_{{\sigma (i_f)},{j_2}}(a_{j_2}) \left(\prod_{l=f}^{k} \bl p_{{\sigma (i_l)},{j_1}}(a_{\sigma (i_l)})\right).
\]
Notice that in the above expression, we first try to explore $a_{j_2}$, and only then $a_{j_1}$. Combining this with Equation \refeq{eq:kijnjvfvs},
{\thinmuskip=.2mu
\medmuskip=0mu plus .2mu minus .2mu
\thickmuskip=1mu plus 1mu
\begin{align*}
&\textnormal{Eq. } \refeq{eq:kijnjvfvs}=\bl p_{i_k,j_1}(a_{i_k})\bl p_{i_k,j_2}(a_{j_2})I_1\\
&\quad + \bl p_{i_k,j_2}(a_{i_k})\left[\prod_{l=2}^{k} \bl p_{{\sigma (i_l)},{j_2}}(a_{\sigma (i_l)})+\sum_{f=2}^{k} \left(\prod_{l=2}^{f-1} \bl p_{{\sigma (i_l)},{j_2}}(a_{\sigma (i_l)})\right) \bl p_{{\sigma (i_f)},{j_2}}(a_{j_2}) \left(\prod_{l=f}^{k} \bl p_{{\sigma (i_l)},{j_1}}(a_{\sigma (i_l)})\right) \right]. \\
&= \bl p_{i_k,j_1}(a_{i_k})\bl p_{i_k,j_2}(a_{j_2})\prod_{l=1}^{k-1} \bl p_{{i_l},{j_1}}(a_{i_l})+\bl p_{i_k,j_2}(a_{i_k})\prod_{l=2}^{k} \bl p_{{\sigma (i_l)},{j_2}}(a_{\sigma (i_l)})\\
&\quad + \bl p_{i_k,j_2}(a_{i_k})\left[\sum_{f=2}^{k} \left(\prod_{l=2}^{f-1} \bl p_{{\sigma (i_l)},{j_2}}(a_{\sigma (i_l)})\right) \bl p_{{\sigma (i_f)},{j_2}}(a_{j_2}) \left(\prod_{l=f}^{k} \bl p_{{\sigma (i_l)},{j_1}}(a_{\sigma (i_l)})\right) \right]\\
&=\bl p_{i_k,j_2}(a_{j_2})\prod_{l=1}^{k} \bl p_{{\sigma(i_l)},{j_1}}(a_{\sigma(i_l)})+\prod_{l=1}^{k} \bl p_{{\sigma (i_l)},{j_2}}(a_{\sigma (i_l)})\\
&\quad + \bl p_{i_k,j_2}(a_{i_k})\left[\sum_{f=2}^{k} \left(\prod_{l=2}^{f-1} \bl p_{{\sigma (i_l)},{j_2}}(a_{\sigma (i_l)})\right) \bl p_{{\sigma (i_f)},{j_2}}(a_{j_2}) \left(\prod_{l=f}^{k} \bl p_{{\sigma (i_l)},{j_1}}(a_{\sigma (i_l)})\right) \right]\\
&=\prod_{l=1}^{k} \bl p_{{\sigma (i_l)},{j_2}}(a_{\sigma (i_l)})+\sum_{f=1}^{k} \left(\prod_{l=1}^{f-1} \bl p_{{\sigma (i_l)},{j_2}}(a_{\sigma (i_l)})\right) \bl p_{{\sigma (i_f)},{j_2}}(a_{j_2}) \left(\prod_{l=f}^{k} \bl p_{{\sigma (i_l)},{j_1}}(a_{\sigma (i_l)})\right), 
\end{align*}}%
and the latter is precisely $\overline Q(\tilde \rho, U)$.
\paragraph{Step 3}
The two previous steps imply that for any $a_{i'}\in \ug, a_{j'}\in \ul$, it holds that
\[
\overline Q_{i',j'}(U)=\overline Q_{i',j_1}(U)= \overline Q_{i_1,j_1}(U).
\]
This completes the proof of Proposition \ref{prop:case of two}.
\end{proofof}
\begin{proofof}{Claim \ref{claim:triplets additional}}
To ease readability, let $\tmu{i}\defeq\abs{\mu(a_i)}$ for every $a_i\in A$. It holds that
\begin{align*}
&\bl p_{{i_2},{j_2}}(a_{i_2})\left( \bl p_{{i_1},{j_1}}(a_{j_1})\bl p_{{i_1},{j_2}}(a_{i_1})+\bl p_{{i_1},{j_1}}(a_{i_1})\bl p_{{i_2},{j_1}}(a_{j_1})\right)\\
&=\frac{\tmu{j_2}}{\tmu{i_2}+\tmu{j_2}}\left( \frac{\tmu{i_1}}{\tmu{i_1}+\tmu{j_1}}\frac{\tmu{j_2}}{\tmu{i_1}+\tmu{j_2}}+\frac{\tmu{j_1}}{\tmu{i_1}+\tmu{j_1}}\frac{\tmu{i_2}}{\tmu{i_2}+\tmu{j_1}}\right) \\
&= \frac{\tmu{j_2}\tmu{i_1}\tmu{j_2}(\tmu{i_2}+\tmu{j_1})+\tmu{j_2}\tmu{j_1}\tmu{i_2}(\tmu{i_1}+\tmu{j_2})}{(\tmu{i_1}+\tmu{j_1})(\tmu{i_1}+\tmu{j_2})(\tmu{i_2}+\tmu{j_1})(\tmu{i_2}+\tmu{j_2})}\\
&= \frac{\overbrace{\tmu{j_2}\tmu{i_1}\tmu{j_2}\tmu{i_2}}^{I}+\overbrace{\tmu{j_2}\tmu{i_1}\tmu{j_2}\tmu{j_1}}^{II}+\overbrace{\tmu{j_2}\tmu{j_1}\tmu{i_2}\tmu{i_1}}^{III}+\overbrace{\tmu{j_2}\tmu{j_1}\tmu{i_2}\tmu{j_2}}^{IV}}{(\tmu{i_1}+\tmu{j_1})(\tmu{i_1}+\tmu{j_2})(\tmu{i_2}+\tmu{j_1})(\tmu{i_2}+\tmu{j_2})}\\
&= \frac{\overbrace{\tmu{j_2}\tmu{i_1}\tmu{j_2}\tmu{i_2}}^{I}+\overbrace{\tmu{j_2}\tmu{j_1}\tmu{i_2}\tmu{j_2}}^{IV}+\overbrace{\tmu{j_2}\tmu{j_1}\tmu{i_2}\tmu{i_1}}^{III}+\overbrace{\tmu{j_2}\tmu{i_1}\tmu{j_2}\tmu{j_1}}^{II}}{(\tmu{i_1}+\tmu{j_1})(\tmu{i_1}+\tmu{j_2})(\tmu{i_2}+\tmu{j_1})(\tmu{i_2}+\tmu{j_2})}\\
&= \frac{\tmu{j_2}\tmu{i_2}\tmu{j_2}(\tmu{i_1}+\tmu{j_1})+\tmu{j_2}\tmu{j_1}\tmu{i_1}(\tmu{i_2}+\tmu{j_2})}{(\tmu{i_1}+\tmu{j_1})(\tmu{i_1}+\tmu{j_2})(\tmu{i_2}+\tmu{j_1})(\tmu{i_2}+\tmu{j_2})}\\
&=\bl p_{{i_1},{j_2}}(a_{i_1})\left(\bl p_{{i_2},{j_1}}(a_{j_1})\bl p_{{i_2},{j_2}}(a_{i_2})+\bl p_{{i_2},{j_1}}(a_{i_2})\bl p_{{i_1},{j_1}}(a_{j_1})  \right).
\end{align*}
\end{proofof}

%% file: input/safe_exploration_theorem.tex
\begin{proofof}{Theorem \ref{thm:holy grail}}
Fix an arbitrary instance. We prove the claim by a two-dimensional induction on the size of $\above(s),\below(s)$, for states $s\in \mS$. The base cases are 
\begin{itemize}
\item $\abs{\above(s)}\geq 2 $ and $\abs{\below(s)} = 1$ (Proposition \ref{prop:W case of one}), and
\item $\abs{\above(s)}=1$ and $\abs{\below(s)}\geq 2$ (Proposition \ref{prop:W case of one strong}),
\end{itemize}
which we relegate to Section \ref{sec:for theorem}. Next, assume the statement holds for all $s\in \mS$ such that $\abs{\above(s)}\leq K_1$, $\abs{\below(s)}\leq K_2$ and $\abs{\above(s)}+\abs{\below(s)}< K_1+K_2$. Let $U\in\mS$ denote a state with $\abs{\above(U)}=K_1$ and $\abs{\below(U)}=K_2$. For abbreviation, let $\ug\defeq\above(U)=\{a_{i_1},a_{i_2},\dots ,a_{i_{K_1}}\}$ and $\ul\defeq\below(U)=\{a_{j_1},a_{j_2},\dots ,a_{j_{K_2}}\}$, and assume the indices follow the stochastic order. Further, for every $a_i \in \ug, a_j\in \ul$ let 
\[
W^*_{i,j}(\ug,\ul)\defeq\bl p_{i,j}(a_j)W^*(\ug,\ul\setminus \{a_j\}) +\bl p_{i,j}(a_i)W^*(\ug\setminus\{a_i\},\ul).
\]
We need to prove that $W^*_{{i_1},{j_1}}(s)=W^*(s)$.
\paragraph{Remarks} Notice that if $X_{a_i'}>0$ for $a_{i'}\in \above(A)$, any $\mP$-valid policy gets $W^*(s)$. To see this, recall that $\mP$-valid policies reach terminate states only after exploring all arms in $\above(A)$. Consequently, we assume for the rest of the proof that $X_{a_i'}\leq 0$ for $a_{i'}\in \above(A)$. 
\paragraph{Step 1}
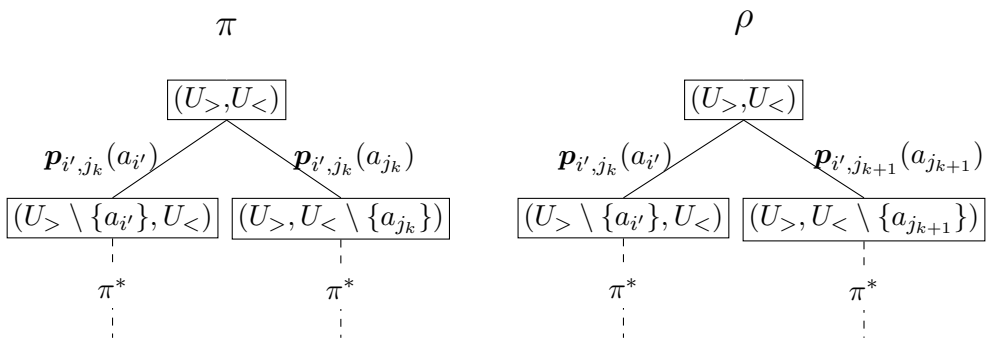
\begin{figure}[htbp]
\centering
\input{input/tree_thm_illustration}
\label{fig:helping for thm two}
\caption{Illustration of the policies $\pi,\rho$ from Step 1 of Theorem \ref{thm:holy grail}. Notice that the left sub-trees of $\pi$ and $\rho$ are identical.
}
\end{figure}

Fix $a_{i'} \in \ug$. We show that for every $k$, $1\leq k<K_2$ it holds that $W^*_{{i'},{j_k}}(s) \geq W^*_{{i'},{j_{k+1}}}(s)$. We define the ordered, $\mP$-valid policy $\pi^*$ by $\sigl_\pi = (a_{i'},\dots)$  namely, $\sigl_{\pi^*}$ ranks $a_{i'}$ first and all the other arms in $\ug$ arbitrarily, and $\sigr_{\pi^*}=(a_{j_1},a_{j_2},\dots,a_{j_{K_2}})$. Due to the inductive assumption, for every state $s$ with $\abs{s}<\abs{U}$, $W^*(s) = W(\pi^*,s)$. Next, we define the policies $\pi,\rho$ explicitly, as follows:
\[
\pi(s) = 
\begin{cases}
\bl p_{{i'},{j_k}} & \textnormal{if $s=(\ug,\ul)$} \\
\pi^*(s)& \textnormal{otherwise} 
\end{cases},\quad 
\rho(s) = 
\begin{cases}
\bl p_{{i'},{j_{k+1}}} & \textnormal{if $s=(\ug,\ul)$} \\
\pi^*(s)& \textnormal{otherwise}
\end{cases}.
\]
We illustrate $\pi$ and $\rho$ in Figure \ref{fig:helping for thm two}. Note that both $\pi,\rho$ have on-path states that are off-path for $\pi^*$. For instance, $\rho$ reaches the state $(\ug,\ul\setminus\{a_{j_{k+1}}\})$ with positive probability, while $\pi^*$ cannot reach it at all. In addition, $\pi$ and $\rho$ are left-ordered with $\sigl_\pi=\sigl_\rho=\sigl_{\pi^*}$. Due to the inductive assumption, it is enough to show that $W(\pi,U)-W(\rho,U)\geq 0$, as this implies $W^*_{{i'},{j_k}}(s) \geq W^*_{{i'},{j_{k+1}}}(s)$.

Next, we factorize $W(\pi,U)$ as follows: for every state $U'\subset U$, we factor $W(\pi,U')$ as long as $a_{j_k},a_{j_{k+1}}\in U'$. Once we reach a term $W(\pi,U')$ with $a_{j_k},a_{j_{k+1}}\notin U'$, we stop. Let $\Psi\defeq \prefix(a_{j_1},a_{j_2}\dots, a_{j_{k-1}})$ be the set of (possibly empty) prefixes of the first $k-1$ arms in $\ul$ according to $\pi^*$. Observe that\footnote{The reader can think of $\psi$ as the set of arms from $(a_{j_1},a_{j_2}\dots, a_{j_{k-1}})$ that were explored.}  
{\thinmuskip=.2mu
\medmuskip=0mu plus .2mu minus .2mu
\thickmuskip=1mu plus 1mu
\begin{align}\label{eq:w of pi}
W(\pi,U)&=\sum_{\psi\in \Psi}
\Pr(\pathto{s}{(\ul\setminus \psi)})R(\ul\setminus \psi)+\Pr(\pathto{s}{(\ul\setminus (\psi \cup \{a_{j_k}\})})R(\ul\setminus (\psi \cup \{a_{j_k}\})\nonumber \\
&\qquad +\underbrace{\sum_{\substack{Z\in \suff(\sigl_\pi)}}
f^\pi_Z \cdot W(\pi,Z,\ul\setminus \{a_{j_1},a_{j_2},\dots,a_{j_{k}},a_{j_{k+1}} \})}_{I^\pi}
\end{align}}
The coefficients $(f^\pi_Z)$ follow from the factorization process. Using similar factorization,
{\thinmuskip=.2mu
\medmuskip=0mu plus .2mu minus .2mu
\thickmuskip=1mu plus 1mu
\begin{align}\label{eq:w of rho}
W(\rho,U)&=\sum_{\psi\in \Psi}
\Pr(\pathtorho{s}{(\ul\setminus \psi)})R(\ul\setminus \psi)+\Pr(\pathtorho{s}{(\ul\setminus (\psi \cup \{a_{j_{k+1}}\}))})R(\ul\setminus (\psi \cup \{a_{j_{k+1}}\}))\nonumber\\
&\qquad +\underbrace{\sum_{\substack{Z\in \suff(\sigl_\rho)}}
f^\rho_Z \cdot W(\rho,Z,\ul\setminus \{a_{j_1},a_{j_2},\dots,a_{j_{k}},a_{j_{k+1}} \})}_{I^\rho}. 
\end{align}}%
Next, we express $(f^\pi_Z)_Z$ in terms of $Q$. By relying on the Equivalence lemma, we show that
\begin{claim}\label{claim:f are equal}
For every $Z\in \suff(\sigl_\pi)=\suff(\sigl_\rho)$, it holds that $f^\pi_Z=f^\rho_Z$.
\end{claim}
To see why Claim \ref{claim:f are equal} holds, notice that we can represent $f^\pi_Z$ as 
{\thinmuskip=.2mu
\medmuskip=0mu plus .2mu minus .2mu
\thickmuskip=1mu plus 1mu
\[
Q(\pi,\ug\setminus Z \cup \{a_{i(Z)}\}, \ul\setminus \{a_{j_1},a_{j_2},\dots,a_{j_{k}},a_{j_{k+1}} \})-Q(\pi,\ug\setminus Z, \ul\setminus \{a_{j_1},a_{j_2},\dots,a_{j_{k}},a_{j_{k+1}} \}),
\]}%
where $a_{i(Z)}$ is the minimal element in $Z$ according to $\sigl_\pi$. By invoking the Equivalence lemma, we can replace $\pi$ in the above expression with $\rho$ and thus for every $Z\in \suff(\sigl_\pi)=\suff(\sigl_\rho)$, it holds that $f^\pi_Z=f^\rho_Z$. Combining Claim \ref{claim:f are equal} with the inductive step and $\sigl_\pi,\sigl_\rho$ being equal, we conclude that $I^\pi=I^\rho$.

Next, we focus on the first sum of $W(\pi,U)$ in Equation \refeq{eq:w of pi}. For every $\psi\in \Psi$, we denote the event $E^\pi_\psi$ as a shorthand for
\[
E^\pi_\psi \defeq \left(\pathto{s}{(\ul\setminus \psi)}\right)\cup\left( \pathto{s}{(\ul\setminus (\psi \cup \{a_{j_k}\}))}\right).
\]
In words, $E^\pi_\psi$ is the event that the GMDP reaches the final state $(\ul\setminus \psi)$ or $(\ul\setminus (\psi \cup \{a_{j_k}\}))$ after starting in $s_0$ and following $\pi$. We use conditional expectation to simplify the summands in the first sum of Equation (\ref{eq:w of pi}),
{\thinmuskip=.2mu
\medmuskip=0mu plus .2mu minus .2mu
\thickmuskip=1mu plus 1mu
\begin{align}\label{eq:pi with alpha}
&\Pr(\pathto{s}{(\ul\setminus \psi)})R(\ul\setminus \psi)+\Pr(\pathto{s}{(\ul\setminus (\psi \cup \{a_{j_k}\})})R(\ul\setminus (\psi \cup \{a_{j_k}\})\nonumber\\
&=\Pr\left(E^\pi_\psi\right)\cdot \underbrace{\left( \Pr(\pathto{s}{(\ul\setminus \psi)}\mid E^\pi_\psi)R(\ul\setminus \psi)+\Pr(\pathto{s}{(\ul\setminus (\psi \cup \{a_{j_k}\})}\mid E^\pi_\psi)R(\ul\setminus (\psi \cup \{a_{j_k}\}))\right)}_{\alpha^\pi_\psi}\nonumber\\
&=\Pr(E^\pi_\psi) \alpha^\pi_\psi. 
\end{align}}%
Similarly for $\rho$, by letting
\[
E^\rho_\psi \defeq \left(\pathtorho{s}{(\ul\setminus \psi)}\right)\cup\left( \pathtorho{s}{(\ul\setminus (\psi \cup \{a_{j_{k+1}}\}))}\right)
\]
and following the derivation in Equation \refeq{eq:pi with alpha} for $\rho$, we get
{\thinmuskip=.2mu
\medmuskip=0mu plus .2mu minus .2mu
\thickmuskip=1mu plus 1mu
\begin{align}\label{eq:rho with alpha}
&\Pr(\pathtorho{s}{(\ul\setminus \psi)})R(\ul\setminus \psi)+\Pr(\pathtorho{s}{(\ul\setminus (\psi \cup \{a_{j_{k+1}}\})})R(\ul\setminus (\psi \cup \{a_{j_{k+1}}\})\nonumber\\
&=\Pr\left(E^\rho_\psi\right)\cdot \underbrace{\left( \Pr(\pathtorho{s}{(\ul\setminus \psi)}\mid E^\rho_\psi)R(\ul\setminus \psi)+\Pr(\pathtorho{s}{(\ul\setminus (\psi \cup \{a_{j_{k+1}}\})}\mid E^\rho_\psi)R(\ul\setminus (\psi \cup \{a_{j_{k+1}}\})) \right)}_{\alpha^\rho_\psi}\nonumber\\
&=\Pr(E^\rho_\psi) \alpha^\rho_\psi. 
\end{align}}%
Due to the fact that $I^\pi=I^\rho$ and relying on Equations (\ref{eq:w of pi})-(\ref{eq:rho with alpha}), we have  
\begin{align}\label{eq:w minus w}
W(\pi,U)-W(\rho,U)=\sum_{\psi\in \Psi}\left(\Pr(E^\pi_\psi) \alpha^\pi_\psi-\Pr(E^\rho_\psi) \alpha^\rho_\psi\right).
\end{align}
At this point, we might be tempted to show that every summand of the sum in Equation \refeq{eq:w minus w} is non-negative. Unfortunately, this approach would not work---stochastic dominance does not mean $\Pr(E^\pi_\psi) \alpha^\pi_\psi\geq \Pr(E^\rho_\psi) \alpha^\rho_\psi$\omer{want to write the example somewhere?}. Instead, we take a different approach. For every $l, 1\leq l \leq k-1$ let $E^\pi_{l}$ denote the event that $a_{j_l}$ was observed in the final state. Formally,
\[
E^\pi_{l} \defeq  \bigcup_{\substack{\psi' \in \Psi\\a_{j_l} \in \psi' }} \left(\pathto{s}{(\ul\setminus \psi')}\right)\cup\left(\pathto{s}{(\ul\setminus (\psi' \cup \{a_{j_k}\})}\right).
\]
In addition, we let $E^\pi_{k}$ denote the empty event, and $E^\pi_{0}$ be the full event (that occurs w.p.~$1$).  \omer{optional: write formally the probability space-in an earlier section} For every non-empty $\psi \in \Psi$, i.e. $\abs{\psi}\geq 1$, let $\max(\psi)\defeq\argmax_{j_l:{a_{j_l}} \in \psi} \sigr_{\pi^*}(a_{j_l})$. According to our assumption about the index of the arms, $\max(\psi)$ is simply the maximal index of an arm in $\psi$ (that is well-defined when $\abs{\psi}\geq 1$). In addition, for completeness, if $\psi=\emptyset$ we let $\max(\emptyset)=0$. We can use the terms $(E^\pi_{l})_{l=0}^k$ to provide an alternative form for $E^\pi_\psi$:
\[
E^\pi_\psi=E^\pi_{\max(\psi)}\setminus E^\pi_{\max(\psi)+1}.
\]
Put differently, for $\psi$ with $\max(\psi)<k-1$, $E^\pi_\psi$ can be viewed as the collection of all of events in which arm $a_{j_{\max(\psi)}}$ was explored, but arm $a_{j_{\max(\psi)+1}}$ was not (where the ``arm'' $a_{j_0}$ for $\psi=\emptyset$ refers to $E^\pi_{0}$). Recall that for $\psi$ with $\max(\psi)=k-1$, $E^\pi_\psi=E^\pi_{\max(\psi)}$ since $E^\pi_{k}$ is the empty event.

Since $E^\pi_{\max(\psi)+1}\subseteq E^\pi_{\max(\psi)}$, we have that for every $\psi\in \Psi$,
\begin{align}\label{eq:e pi with Z}
\Pr(E^\pi_\psi)=\Pr(E^\pi_{\max(\psi)})-\Pr(E^\pi_{\max(\psi)+1}),
\end{align}
taking care of edge cases too. By renaming $\alpha^\pi_\psi$ to $\alpha^\pi_{\max(\psi)}$, i.e., $\alpha^\pi_{\max(\psi)} \defeq \alpha^\pi_\psi $, and rearranging Equation \refeq{eq:w of pi} using Equations \refeq{eq:pi with alpha} and \refeq{eq:e pi with Z},
{\thinmuskip=.2mu
\medmuskip=0mu plus .2mu minus .2mu
\thickmuskip=1mu plus 1mu
\begin{align}\label{eq:w pi with diff}
W(\pi,U)&= I^\pi+\sum_{\psi\in \Psi} \Pr\left(E^\pi_\psi\right) \alpha^\pi_\psi=I^\pi+\sum_{l=0}^{k-1}\left(\Pr(E^\pi_{l})-\Pr(E^\pi_{l+1})\right) \alpha^\pi_l.
\end{align}}%
By defining $E^\rho_l$ analogously,
\[
E^\rho_{l} \defeq  \bigcup_{\substack{\psi' \in \Psi\\a_{j_l} \in \psi' }} \left(\pathtorho{s}{(\ul\setminus \psi')}\right)\cup\left(\pathtorho{s}{(\ul\setminus (\psi' \cup \{a_{j_{k+1}}\})}\right),
\]
and following similar arguments, we conclude that
{\thinmuskip=.2mu
\medmuskip=0mu plus .2mu minus .2mu
\thickmuskip=1mu plus 1mu
\begin{align}\label{eq:w rho with diff}
W(\rho,U)&= I^\rho+\sum_{\psi\in \Psi} \Pr\left(E^\rho_\psi\right) \alpha^\rho_\psi=I^\rho+\sum_{l=0}^{k-1}\left(\Pr(E^\rho_{l})-\Pr(E^\rho_{l+1})\right) \alpha^\rho_l.
\end{align}}%
By rephrasing Equation (\ref{eq:w minus w}) using Equations (\ref{eq:w pi with diff}) and (\ref{eq:w rho with diff}),
\begin{align}\label{eq: thm step 1 good}
W(\pi,U)-W(\rho,U)=\sum_{l=0}^{k-1}\left(\Pr(E^\pi_{l})-\Pr(E^\pi_{l+1})\right) \alpha^\pi_l-\left(\Pr(E^\rho_{l})-\Pr(E^\rho_{l+1})\right) \alpha^\rho_l.
\end{align}
Next, we show two monotonicity properties.
\begin{proposition}\label{prop: monotonicity in thm}
Under Assumption \ref{assumption:dominance},
\begin{enumerate}
\item for every $l\in \{0,1,\dots,k-1 \}$, it holds that $\alpha^\pi_l \geq \alpha^\rho_l$. \label{item:prop alphas rho pi}
\item for every $l\in \{0,1,\dots,k-2\}$, it holds that $\alpha^\pi_{l+1} \geq \alpha^\pi_{l}$ and $\alpha^\rho_{l+1} \geq \alpha^\rho_{l}$. \label{item:prop alphas alpha pi}
\end{enumerate}
\end{proposition}
In fact, this is the only place in the proof of Theorem \ref{thm:holy grail} where we rely on Assumption \ref{assumption:dominance}. Equipped with Proposition \ref{prop: monotonicity in thm}, we can make the final argument. For every $r,r\in \{1\dots,k-1\}$ let
\[
f(r) \defeq \left( \Pr(E^\pi_{r})-\Pr(E^\rho_{r}) \right)\alpha^\pi_{r-1}.
\]
In addition, let
\[
g(r)\defeq \sum_{l=0}^{r}\left(\Pr(E^\pi_{l})-\Pr(E^\pi_{l+1})\right) \alpha^\pi_l-\left(\Pr(E^\rho_{l})-\Pr(E^\rho_{l+1})\right) \alpha^\rho_l
\]
We shall show that for every $r,r\in \{0,\dots,k-2\}$ it holds that
\begin{align}\label{eq:thm last argument}
W(\pi,U)-W(\rho,U)\geq f(r+1)+g(r).
\end{align}
For $r=k-2$, we have
{
\begin{align*}
\textnormal{Eq. (\ref{eq: thm step 1 good})}&=\left(\Pr(E^\pi_{k-1})-\Pr(E^\pi_{k})\right) \alpha^\pi_{k-1}-\left(\Pr(E^\rho_{k-1})-\Pr(E^\rho_{k})\right) \alpha^\rho_{k-1}+g(k-2)\\
&\stackrel{\substack{E^\pi_{k},E^\rho_{k}\\\textnormal{are empty}}}{=}\Pr(E^\pi_{k-1})\alpha^\pi_{k-1}-\Pr(E^\rho_{k-1}) \alpha^\rho_{k-1}+g(k-2)\\
&\stackrel{\textnormal{Prop. \ref{prop: monotonicity in thm}.\ref{item:prop alphas rho pi}}}{\geq}\left(\Pr(E^\pi_{k-1})-\Pr(E^\rho_{k-1}) \right)\alpha^\pi_{k-1}+g(k-2)\\
&\stackrel{\textnormal{Prop. \ref{prop: monotonicity in thm}.\ref{item:prop alphas alpha pi}}}{\geq}\left(\Pr(E^\pi_{k-1})-\Pr(E^\rho_{k-1}) \right)\alpha^\pi_{k-2}+g(k-2)\\
&=f(k-1)+g(k-2).
\end{align*}}%
Assume Inequality \refeq{eq:thm last argument} holds for $r+1$. Then, for $r$ we have 
{\thinmuskip=.2mu
\medmuskip=0mu plus .2mu minus .2mu
\thickmuskip=1mu plus 1mu
\begin{align*}
f(r+1)+g(r) &= \left( \Pr(E^\pi_{r+1})-\Pr(E^\rho_{r+1}) \right)\alpha^\pi_{r}+g(r-1)\\
&\qquad \qquad +\left(\Pr(E^\pi_{r})-\Pr(E^\pi_{r+1})\right) \alpha^\pi_r-\underbrace{\left(\Pr(E^\rho_{r})-\Pr(E^\rho_{r+1})\right)}_{\geq 0,\textnormal{ Eq. \refeq{eq:e pi with Z}}}  \alpha^\rho_r\\
&\stackrel{\textnormal{Prop. \ref{prop: monotonicity in thm}.\ref{item:prop alphas rho pi}}}{\geq} \left( \Pr(E^\pi_{r+1})-\Pr(E^\rho_{r+1}) \right)\alpha^\pi_{r}+g(r-1)\\
&\qquad \qquad +\left(\Pr(E^\pi_{r})-\Pr(E^\pi_{r+1})\right) \alpha^\pi_r-\left(\Pr(E^\rho_{r})-\Pr(E^\rho_{r+1})\right) \alpha^\pi_r\\
&=\Pr(E^\pi_{r})\alpha^\pi_r-\Pr(E^\rho_{r}) \alpha^\pi_r +g(r-1)\\
&\stackrel{\textnormal{Prop. \ref{prop: monotonicity in thm}.\ref{item:prop alphas alpha pi}}}{\geq}\Pr(E^\pi_{r})\alpha^\pi_{r-1}-\Pr(E^\rho_{r}) \alpha^\pi_{r-1} +g(r-1)\\
&= f(r)+g(r-1).
\end{align*}}%
Ultimately, by setting $r=0$ in Inequality \refeq{eq:thm last argument},
\begin{align*}
W(\pi,U)-W(\rho,U)&\geq f(1)+g(0)\\
&=\left( \Pr(E^\pi_{1})-\Pr(E^\rho_{1}) \right)\alpha^\pi_{0}\\
&\qquad \qquad +\left(\Pr(E^\pi_{0})-\Pr(E^\pi_{1})\right) \alpha^\pi_0-\left(\Pr(E^\rho_{0})-\Pr(E^\rho_{1})\right) \alpha^\rho_0\\
&\geq\Pr(E^\pi_{0}) \alpha^\pi_0-\Pr(E^\rho_{0}) \alpha^\pi_0\\
&=0.
\end{align*}
This concludes the first step of the theorem.
\paragraph{Step 2}  In this step, we show that for every $k$, $1\leq k < K_1$ it holds that $W^*_{{i_k},{j_1}}(s) = W^*_{{i_{k+1}},{j_1}}(s)$. Define an ordered, $\mP$-valid policy $\pi^*$ by $\sigl_\pi = (a_{i_1},a_{i_2},\dots,a_{i_{K_1}})$,  namely, $\sigl_{\pi^*}$ ranks the elements of $\ug$ according to the stochastic order, and $\sigr_{\pi^*}=(a_{j_1},a_{j_2},\dots,a_{j_{K_2}})$. Due to the inductive assumption, for every state $s$ with $\abs{s}<\abs{U}$, $W^*(s) = W(\pi^*,s)$. Next, we define the policies $\pi,\rho$ explicitly, as follows:
\[
\pi(s) = 
\begin{cases}
\bl p_{{i_k},{j_1}} & \textnormal{if $s=(\ug,\ul)$} \\
\pi^*(s)& \textnormal{otherwise} 
\end{cases},\quad 
\rho(s) = 
\begin{cases}
\bl p_{{i_{k+1}},{j_{1}}} & \textnormal{if $s=(\ug,\ul)$} \\
\pi^*(s)& \textnormal{otherwise}
\end{cases}.
\]
As in the previous step, the inductive step suggests that showing $W(\pi,s)=W(\rho,s)$ is suffice. However, unlike the previous step, here the set of reachable terminal state is the same for $\pi$ and $\rho$; hence, this equality is almost immediate due to the Equivalence lemma. Let $\Psi'\defeq \prefix(a_{j_1},a_{j_2}\dots, a_{j_{K_2-1}})$ be the set of (possibly empty) prefixes of the arms in $\ul \setminus \{a_{j_{K_2}}\}$ according to $\pi^*$. Observe that we can factorize $W(\pi,s)$ as follows:
\begin{align}\label{eq:w for pi with prob}
W(\pi,s) &= Q(\pi,U)\cdot R(\emptyset)+ \sum_{\psi \in \Psi}\Pr(\pathto{s}{(\ul \setminus \psi)})R(\ul \setminus \psi).
\end{align}
The next Claim \ref{claim:thm:step 2 claim} suggests we can replace probabilities with functions of $Q$.
\begin{claim}\label{claim:thm:step 2 claim}
For every $\psi \in \Psi$, it holds that 
\[
\Pr(\pathto{s}{(\ul \setminus \psi)})=Q(\pi,\ug, \psi)-Q(\pi,\ug,\psi \cup\{a_{j_{max(\psi)+1}} \}).
\]
\end{claim}
By applying the Equivalence lemma on the statement of  Claim \ref{claim:thm:step 2 claim}, we obtain that for every $\psi \in \Psi$
\[
\Pr(\pathto{s}{(\ul \setminus \psi)})=\Pr(\pathtorho{s}{(\ul \setminus \psi)});
\]
hence, we can rewrite Equation \refeq{eq:w for pi with prob} as 
\begin{align*}
W(\pi,s) &= Q(\rho,U)\cdot R(\emptyset)+ \sum_{\psi \in \Psi}\Pr(\pathtorho{s}{(\ul \setminus \psi)})R(\ul \setminus \psi)\\
&=W(\rho,s).
\end{align*}
This concludes the second step of the theorem.
\paragraph{Step 3 (final)} We are ready to prove the theorem. Fix arbitrary $a_{\tilde i}$ and $a_{\tilde j}$ such that $a_{\tilde i} \in \ug$ and  $a_{\tilde j} \in \ul$. By the previous steps, we know that
\[
W^*_{{i_1},{j_1}}(U)\stackrel{\textnormal{Step 2}}{=}W^*_{{\tilde i},{j_1}}(U)\stackrel{\textnormal{Step 1}}{\geq}W^*_{{\tilde i},{\tilde j}}(U).
\]
This ends the proof of Theorem \ref{thm:holy grail}.
\end{proofof}

%% file: input/tree_thm_illustration.tex
\forestset{
 strongedge label/.style 2 args={
    edge label={node[midway,left, #1]{#2}},
  }, 
 weakedge label/.style 2 args={
    edge label={node[midway,right, #1]{#2}},
  }, 
   straightedge label/.style 2 args={
    edge label={node[midway, #1]{#2}},
  }, 
  important/.style={draw={red,thick,fill=red}}
}
\begin{forest} 
[{\Large $\pi$},
 [{$(U_>{,}U_<)$}, edge={white},l*=.05, for tree=
	{
	draw,
	font=\sffamily,
	l+=.5cm,
	inner sep=2pt,
	l sep=5pt,
	s sep=5pt,
	parent anchor=south,
	child anchor=north
    }
 	[${(U_>\setminus\{ a_{i'}\},U_<)}$, strongedge label={left}{$\bl p_{{i'},{j_k}}(a_{i'})$}
 		[,white, edge={dashed}, straightedge label={fill=white}{$\pi^*$}
 		]
 	]
	[${(U_>,U_<\setminus\{ a_{j_k}\})}$, weakedge label={right}{$\bl p_{{i'},{j_k}}(a_{j_k})$}
 		[,white, edge={dashed}, straightedge label={fill=white}{$\pi^*$}
 		]
 	]
 ]
]
\end{forest}
\qquad
\begin{forest} 
[{\Large $\rho$},
 [{$(U_>{,}U_<)$}, edge={white},l*=.05, for tree=
	{
	draw,
	font=\sffamily,
	l+=.5cm,
	inner sep=2pt,
	l sep=5pt,
	s sep=5pt,
	parent anchor=south,
	child anchor=north
    }
 	[${(U_>\setminus\{ a_{i'}\},U_<)}$, strongedge label={left}{$\bl p_{{i'},{j_k}}(a_{i'})$}
 		[,white, edge={dashed}, straightedge label={fill=white}{$\pi^*$}
 		]
 	]
	[${(U_>,U_<\setminus\{ a_{j_{k+1}}\})}$, weakedge label={right}{$\bl p_{{i'},{j_{k+1}}}(a_{j_{k+1}})$}
 		[,white, edge={dashed}, straightedge label={fill=white}{$\pi^*$}
 		]
 	]
 ]
]
\end{forest}

%% file: input/safe_exploration_thm_additions.tex
\begin{proposition}\label{prop:W case of one}
Let $U$ be a state such that $\abs{\above(U)}\geq 2$ and $\abs{\below(U)} =1$. 
It holds that $W(\pi^*,U)=W^*(U)$.
\end{proposition}
\begin{proofof}{Proposition \ref{prop:W case of one}}
Denote $\below(U)=\{a_{j}\}$. We can assume w.l.o.g. that the realization of all arms in $\above(U)$ are non-positive, as otherwise every $\mP$-valid policy will explore all the arms; thus, $W^*(U)=Q^*(U)\cdot \max\{0, X_{a_{j}}  \}$. Finally, the Equivalence lemma suggests that   $Q^*(U)$ is policy invariant; hence, $W(\pi,U)=W^*(U)$ holds for any $\mP$-valid policy $\pi$, and in particular for $\pi=\pi^*$.
\end{proofof}

\begin{proposition}\label{prop:W case of one strong}
Let $U$ be a state such that $\abs{\above(U)}=1$ and $\abs{\below(U)} \geq 2$. It holds that $W(\pi^*,U)=W^*(U)$.
\end{proposition}
\begin{proofof}{Proposition \ref{prop:W case of one strong}}
This statement is a special case of Proposition \ref{prop:index with ugeq one} for instances satisfying Assumption \ref{assumption:dominance}.
\end{proofof}
\begin{proposition}\label{prop:index with ugeq one}
Let $U$ be a state such that $\abs{\above(U)}=1$ and $\abs{\below(U)} \geq 2$. Let $f^*$ be a real-valued function, $f^*:\below(U)\rightarrow \mathbb R$, such that for every $a_l \in \below(U)$,
\[
f^*(a_{l})=\frac{\Pr(X_{a_l} > 0 )\E(\max_{a_j \in {U} } X_{a_j}\mid X_{a_l}>0 )}{\abs{\mu(a_l)}}.
\]  
Denote by $\pi_{f^*}$ the right-ordered policy that orders $\below(U)$ according to decreasing order of $f^*$. Then, $W(\pi_{f^*},U)=W^*(U)$. 
\end{proposition}

\begin{proofof}{Proposition \ref{prop:index with ugeq one}}
Denote $\above(U)=\ug=\{a_{i}\}$  and $\below(U)=\ul=\{a_{j_1},\dots a_{j_k}\}$ for $k=\abs{\ul}$. Let $\pi$ be any right-ordered policy with the matching $\sigr_\pi$, such that $\pi\neq \pi_{f^*}$. Assume that there are indices $r,l$, for $1\leq r,l \leq k$, such that $\sigr_\pi(a_{l})<\sigr_\pi(a_{r})$ yet $f(a_{l})<f(a_{{r}})$, for arms $a_l,a_r\in \ul$. Moreover, if such a pair $(l,r)$ exists,  assume w.l.o.g. that $\sigr_\pi$ orders them consequentially, i.e., for every arm $a\in U$ such that $a\notin \{ a_{l}, a_{r} \}$,  either $\sigr_\pi(a)< \sigr_\pi(a_{l})$ or $\sigr_\pi(a)> \sigr_\pi(a_{r})$.

Denote by $\pi'$ the right-ordered policy that swaps $a_{l}$ and $a_{r}$. If we show that $\pi'$ yields a better reward than $\pi$, we could swap the order of $\pi$ one pair at a time, thereby showing that $\pi_{f^*}$ is indeed optimal. To simplify notation, we denote by $\ind_{a}$ the event that $X_a>0$ for arm $a\in U$.

Let $\ul'$ be the set of all arms in $\ul$ such that $\ul'=\{a\in \ul\mid \sigr_\pi(a)< \sigr_\pi(a_{l})  \}$. We divide the analysis into two cases: in case $X_{a_{i}} >0$ or $\max_{a\in \ul'} X_a > 0$, both $\pi,\pi'$ obtain the same reward. Otherwise, assume that $X_{a_{i}} \leq 0$ and $\max_{a\in \ul'} X_a \leq 0$ ; hence, in case arm $a_{i}$ is selected, the policy reaches a terminal state with a reward of 0. The reward of $\pi$ is given by 
{\small
\thinmuskip=.2mu
\medmuskip=0mu plus .2mu minus .2mu
\thickmuskip=1mu plus 1mu
\begin{align*}
W(\pi,U)=C_1 \left( \bl p_{i,l}(l)\E(\ind_l \max_{a_j \in {U_< \setminus \ul'}} X_{a_j} )+\bl p_{i,l}(l)\bl p_{i,r}(r)\E((1-\ind_l)\ind_r \max_{a_j \in {U_< \setminus \ul'}} X_{a_j})+\bl p_{i,l}(l)\bl p_{i,r}(r)C_2 \right),
\end{align*}}%
where $C_1$ and $C_2$ are constants that depend on $\sigr_\pi$. Similarly, the reward of $\pi'$ is given by 
{\small
\thinmuskip=.2mu
\medmuskip=0mu plus .2mu minus .2mu
\thickmuskip=1mu plus 1mu
\begin{align*}
W(\pi',U) = C_1\left( \bl p_{i,r}(r)\E(\ind_r \max_{a_j \in {U_< \setminus \ul'}} X_{a_j} )+\bl p_{i,r}(r)\bl p_{i,l}(l)\E((1-\ind_r)\ind_l \max_{a_j \in {U_< \setminus \ul'}} X_{a_j})+\bl p_{i,r}(r)\bl p_{i,l}(l)C_2 \right),
\end{align*}}%
where $C_1$ and $C_2$ are the same constants.
If $W(\pi,U)\geq W(\pi',U)$, then
\begin{align*}
& \bl p_{i,l}(l)\E(\ind_l \max_{a_j \in {U_< \setminus \ul'}} X_{a_j} )+\bl p_{i,l}(l)\bl p_{i,r}(r)\E((1-\ind_l)\ind_r \max_{a_j \in {U_< \setminus \ul'}} X_{a_j}) \\
&\qquad \geq \bl p_{i,r}(r)\E(\ind_r \max_{a_j \in {U_< \setminus \ul'}} X_{a_j} )+\bl p_{i,r}(r)\bl p_{i,l}(l)\E((1-\ind_r)\ind_l \max_{a_j \in {U_< \setminus \ul'}} X_{a_j}),
\end{align*}
implying that
\begin{align*}
&\bl p_{i,l}(l)\E(\ind_l \max_{a_j \in  {U_< \setminus \ul'}} X_{a_j} )(1-\bl p_{i,r}(r)) \geq \bl p_{i,r}(r)\E(\ind_r \max_{a_j \in  {U_< \setminus \ul'}} X_{a_j} )(1-\bl p_{i,l}(l)).
\end{align*}
Stated otherwise,
\begin{align*}
\frac{\bl p_{i,l}(l)\E(\ind_l \max_{a_j \in {U_< \setminus \ul'}} X_{a_j} )}{ 1-\bl p_{i,l}(l)} > \frac{\bl p_{i,r}(r)\E(\ind_r \max_{a_j \in {U_< \setminus \ul'}} X_{a_j} )}{1-\bl p_{i,r}(r) }.
\end{align*}
Finally, due to the definitions of $\bl p_{i,r},\ind_l$ and $\bl p_{i,l}, \ind_r$,
\begin{align*}
\frac{\Pr(X_{a_l} > 0 )\E(\max_{a_j \in {U_< \setminus \ul'} } X_{a_j}\mid X_{a_l}>0 )}{\abs{\mu(a_l)}} \geq \frac{\Pr(X_{a_l} > 0 )\E( \max_{a_j \in {U_< \setminus \ul'}} X_{a_j} \mid X_{a_r} >0 )}{\abs{\mu(a_r)}},
\end{align*}
which contradicts our assumption that $f(a_{l})<f(a_{r})$.
\end{proofof}

\begin{proofof}{Claim \ref{claim:f are equal}}
Notice that we can represent $f^\pi_Z$ as 
{\thinmuskip=.2mu
\medmuskip=0mu plus .2mu minus .2mu
\thickmuskip=1mu plus 1mu
\[
Q(\pi,\ug\setminus Z \cup \{a_{i(Z)}\}, \ul\setminus \{a_{j_1},a_{j_2},\dots,a_{j_{k}},a_{j_{k+1}} \})-Q(\pi,\ug\setminus Z, \ul\setminus \{a_{j_1},a_{j_2},\dots,a_{j_{k}},a_{j_{k+1}} \}),
\]}%
where $a_{i(Z)}$ is the minimal element in $Z$ according to $\sigl_\pi$. This process is similar in spirit to the proof of Proposition \ref{prop:coef c} and is hence omitted. Then, we can mirror the same arguments for $f^\rho_Z$. Finally, the Equivalence lemma suggests that the two representations are equal.
\end{proofof}

\begin{proofof}{Proposition \ref{prop: monotonicity in thm}}
We address the two parts separately below.
\paragraph{Part \ref{item:prop alphas rho pi}} 
Notice that the terminal state $(\ul\setminus \psi)$ is reachable from the left sub-tree of $\pi$ and $\rho$ solely (see Figure \ref{fig:helping for thm two} for illustration). Due to the construction of $\pi$ and $\rho$,
\begin{align}\label{eq:pi to rho}
\Pr(\pathto{s}{(\ul\setminus \psi)})&= \bl p_{{i_1},{j_k}}(a_{i_1})\Pr(\pathto{s\setminus \{a_{i_1}\}}{(\ul\setminus \psi)})\nonumber\\
&=\frac{-\mu(a_{j_{k}})}{- \mu(a_{j_{k}})+ \mu(a_{i_{1}})} \Pr(\pathto{s\setminus \{a_{i_1}\}}{(\ul\setminus \psi)})\nonumber\\
&=\frac{-\mu(a_{j_{k}})}{- \mu(a_{j_{k}})+ \mu(a_{i_{1}})} \Pr(\pathtorho{s\setminus \{a_{i_1}\}}{(\ul\setminus \psi)})\nonumber\\
&\leq\frac{-\mu(a_{j_{k+1}})}{- \mu(a_{j_{k+1}})+ \mu(a_{i_{1}})} \Pr(\pathtorho{s\setminus \{a_{i_1}\}}{(\ul\setminus \psi)})\nonumber\\
&\leq \bl p_{{i_1},{j_{k+1}}}(a_{i_1})\Pr(\pathtorho{s\setminus \{a_{i_1}\}}{(\ul\setminus \psi)})\nonumber \\
& = \Pr(\pathtorho{s}{(\ul\setminus \psi)}),
\end{align}
since $\mu(a_{j_{k+1}})\geq \mu(a_{j_{k}})$, and due to monotonicity of $f(x)=\frac{x}{x+c}$ for positive $c$. Using similar arguments, 
\begin{align}\label{eq: pi with rho again}
\Pr(\pathto{s}{(\ul\setminus (\psi \cup \{a_{j_k}\})}) \geq \Pr(\pathtorho{s}{(\ul\setminus (\psi \cup \{a_{j_{k+1}}\})}).
\end{align}
Now, observe that
\begin{align}\label{eq: with Ez}
\Pr(\pathto{s}{(\ul\setminus \psi)}\mid E^\pi_\psi)
&=\frac{\Pr(\pathto{s}{(\ul\setminus \psi)})}{\Pr(\pathto{s}{(\ul\setminus \psi)}) +\Pr(\pathto{s}{(\ul\setminus (\psi \cup \{a_{j_k}\})})} \nonumber\\
&\stackrel{\textnormal{Eq. }(\ref{eq:pi to rho}),(\ref{eq: pi with rho again})}{\leq} \frac{\Pr(\pathtorho{s}{(\ul\setminus \psi)})}{\Pr(\pathtorho{s}{(\ul\setminus \psi)}) +\Pr(\pathtorho{s}{(\ul\setminus (\psi \cup \{a_{j_{k+1}}\})})} \nonumber\\
&=\Pr(\pathtorho{s}{(\ul\setminus \psi)}\mid E^\rho_\psi),
\end{align}
where the second to last step follows again from monotonicity of $f(x)=\frac{x}{x+c}$ for positive $c$. Further, due to monotonicity of the reward function $R$,
\[
R(\ul\setminus (\psi\cup \{a_{j_k}\})) \geq R(\ul\setminus \psi), \qquad R(\ul\setminus (\psi\cup \{a_{j_{k+1}}\})) \geq R(\ul\setminus \psi).
\]
In addition, due to Assumption \ref{assumption:dominance}, $R(\ul\setminus (\psi\cup \{a_{j_k}\})) \geq  R(\ul\setminus (\psi\cup \{a_{j_{k+1}}\}))$. Wrapping up,
{\thinmuskip=.2mu
\medmuskip=0mu plus .2mu minus .2mu
\thickmuskip=1mu plus 1mu
\begin{align*}
\alpha^\pi_\psi&=\Pr(\pathto{s}{(\ul\setminus \psi)}\mid E^\pi_\psi)R(\ul\setminus \psi)+(1-\Pr(\pathto{s}{(\ul\setminus \psi)}\mid E^\pi_\psi))R(\ul\setminus (\psi \cup \{a_{j_k}\}))\\
&\geq \Pr(\pathto{s}{(\ul\setminus \psi)}\mid E^\pi_\psi)R(\ul\setminus \psi)+(1-\Pr(\pathto{s}{(\ul\setminus \psi)}\mid E^\pi_\psi))R(\ul\setminus (\psi \cup \{a_{j_{k+1}}\}))\\
&\stackrel{\textnormal{Eq. }(\ref{eq: with Ez})}{\geq} \Pr(\pathtorho{s}{(\ul\setminus \psi)}\mid E^\rho_\psi)R(\ul\setminus \psi)+(1-\Pr(\pathtorho{s}{(\ul\setminus \psi)}\mid E^\rho_\psi))R(\ul\setminus (\psi \cup \{a_{j_{k+1}}\}))\\
&=\alpha^\rho_\psi.
\end{align*}}%
This completes the proof of the first part.

\paragraph{Part \ref{item:prop alphas alpha pi}} We prove the claim for $\alpha^\rho_{l+1} \geq \alpha^\rho_{l}$, and the other part is symmetrical. Let $\psi=\psi(l)$ such that $\max(\psi)=l$. Notice that the reward function $R$ is a set function, and is, by definition monotonically decreasing; hence, $R(\ul\setminus \psi) \leq R(\ul\setminus (\psi \cup \{a_{j_{k+1}}\}))$. Consequently,
{\thinmuskip=.2mu
\medmuskip=0mu plus .2mu minus .2mu
\thickmuskip=1mu plus 1mu
\begin{align*}
\alpha^\rho_\psi&=\Pr(\pathtorho{s}{(\ul\setminus \psi)}\mid E^\rho_\psi)R(\ul\setminus \psi)+\Pr(\pathtorho{s}{(\ul\setminus (\psi \cup \{a_{j_{k+1}}\})}\mid E^\rho_\psi)R(\ul\setminus (\psi \cup \{a_{j_{k+1}}\}))\nonumber\\
&\leq R(\ul\setminus (\psi \cup \{a_{j_{k+1}}\})).
\end{align*}}%
Further, let $\psi'$ such that $\max(\psi')=l+1$, namely $\psi'=\psi\cup\{a_{j_{l+1}} \}$. Using the same properties of $R$, we have that $R(\ul\setminus \psi') \leq R(\ul\setminus (\psi' \cup \{a_{j_{k+1}}\}))$; thus,
{\thinmuskip=.2mu
\medmuskip=0mu plus .2mu minus .2mu
\thickmuskip=1mu plus 1mu
\begin{align*}
\alpha^\rho_{\psi'}&=\Pr(\pathtorho{s}{(\ul\setminus \psi')}\mid E^\rho_{\psi'})R(\ul\setminus \psi')+\Pr(\pathtorho{s}{(\ul\setminus (\psi' \cup \{a_{j_{k+1}}\})}\mid E^\rho_{\psi'})R(\ul\setminus (\psi' \cup \{a_{j_{k+1}}\}))\nonumber\\
&\geq R(\ul\setminus \psi' ).
\end{align*}}%
Next, let $V_\psi$ denote the event that $(X_{a})_{a\in \psi}$ attain value below $\alpha$. Observe that 
{\thinmuskip=.2mu
\medmuskip=0mu plus .2mu minus .2mu
\thickmuskip=1mu plus 1mu
\begin{align*}
R(\ul\setminus (\psi \cup \{a_{j_{k+1}}\})) &=\Pr(V_\psi)R(\ul\setminus (\psi \cup \{a_{j_{k+1}}\}))+(1-\Pr(V_\psi))R(\ul\setminus (\psi \cup \{a_{j_{k+1}}\}))\nonumber\\
&=\Pr(V_\psi)\max\{\alpha,X_{a_{j_{k+1}}}\}+(1-\Pr(V_\psi))\max\{a_{j_1},\dots ,a_{j_l},a_{j_{k+1}}  \}\nonumber\\
&\leq \Pr(V_\psi)\max\{\alpha,X_{a_{j_{l+1}}}\}+(1-\Pr(V_\psi))\max\{a_{j_1},\dots ,a_{j_l},a_{j_{l+1}}  \}\\
&=R(\ul\setminus \psi' ) ,
\end{align*}}%
where the second to last inequality is due to Assumption \ref{assumption:dominance} and independence of $(X_{a_i})_{i=1}^K$. Ultimately, 
\[
\alpha^\rho_{l+1}=\alpha^\rho_{\psi'} \geq R(\ul\setminus \psi' ) \geq R(\ul\setminus (\psi \cup \{a_{j_{k+1}}\})) \geq \alpha^\rho_{\psi}=\alpha^\rho_{l}.
\]
This completes the proof of the second part.
\end{proofof}

\begin{proofof}{Claim \ref{claim:thm:step 2 claim}}
The proof goes along the lines of Claim \ref{claim:f are equal}, and is hence omitted.
\end{proofof}

%% file: input/body-proofs.tex
\begin{proofof}{Observation \ref{obs:eventually will explore}}
Let $x(a_i)>0$ for some $i\in [K]$, let $j$ be an index of unexplored arm, and let $\mI$ be the information of the algorithm. We overload the notation $\bl p_{i,j}$ to acknowledge the realized value $x(a_i)$; that is,
\begin{align*}
\bl p_{i,j}(a) =
\begin{cases}
\frac{-\mu(a_j)}{x(a_i)-\mu(a_j)} & \textnormal{if } a=a_i\\
\frac{x(a_i)}{x(a_i)-\mu(a_j)} & \textnormal{if } a=a_j\\
0 & \textnormal{otherwise}
\end{cases}.
\end{align*}
Notice that 
\begin{align*}
\sum_{a\in A}\bl p(a)\E\left[X(a)\mid \mI\right] &= \bl  p_{i,j}(a_i)x(a_i)  + \bl  p_{i,j}(a_j)\mu(a_j) \\
&= x(a_i)\cdot \frac{-\mu(a_j)}{x(a_i)-\mu(a_j)} + \mu(a_j)\cdot \frac{x(a_i)}{x(a_i)-\mu(a_j)} = 0;  
\end{align*}
hence, $\bl p_{i,j}$ is safe w.r.t. to $\mI$. After playing $\bl p_{i,j}$, either $a_i$ was realized or $a_j$. In the former, the information remains the same, and we can repeat this experiment again. The probability of $a_j$ realizing is positive and constant, and hence after finite time we will eventually realize it. Once we do, the number of unexplored armed decreases by one. We can follow this process until all arms are explored.
\end{proofof}

\begin{proofof}{Observation~\ref{obs: U leq W*}}
The proof of this observation relies on constructing a policy $\pi$ that simulates $\ALG$. Since by definition $W(\pi,A) \leq  W^*(A)$, it is enough to show that $\lim_{T \rightarrow \infty }\mU_T(\ALG) \leq W(\pi,A)$. In every round, $\pi$ plays precisely what $\ALG$ plays, and if the realized arm was already explored by $\ALG$, $\pi$ ignores it. The infinite time expected value of $\ALG$ cannot exceed $ W(\pi,A)$. The full details are similar to \cite[Theorem 3]{Fiduciary} and are hence omitted. 
\end{proofof}

\begin{proofof}{Observation \ref{obs: U get W}}
Fix any policy $\pi$. Let $\ALG(\pi)$ be the modification of Algorithm~\ref{alg:alg of pi} that uses $\pi$ instead of $\OGP$ in Lines \ref{algpi:while}-\ref{algpi:play with ogp}. Once $\pi$ reaches a terminal state, $\ALG(\pi)$ secures the reward of $\pi$ in finite time. Overall, $\lim_{T \rightarrow \infty }\mU_T(\ALG(\pi)) = W(\pi,A)$.
\end{proofof}

\begin{proofof}{Proposition \ref{prop:bernoulli opt}}
Fix an $\ise$ instance such that $(X(a_i))_i \in \{x^-,x^+\}$ (for $x^- \leq x^+$) almost surely. For the problem to be non-trivial, we must have $x^- <0$ and $x^+ >0$. Otherwise, if  $x^-,x^+ <0$ the only safe action is $a_0$, and if $x^-,x^+ \geq 0$, we can explore all arms using the singleton portfolios $(\bl p_{ii})_{i \in [K]}$. From here on, we assume w.l.o.g. that $x^- = -1$ and $x^+=H$. For convenience, we  state $\SEGB'$ explicitly in Algorithm~\ref{alg:alg of pi two supported}. Before we prove the proposition, we remark that
\begin{enumerate}
    \item Since $(X(a_i))_i$ take either $-1$ or $H$, Assumption~\ref{assumption:dominance} implies a stochastic order on all arms, not only on $\below(A)$. 
    \item Any asymptotically optimal algorithm conducts at most $K$ exploration rounds before it exploits. This implies an immediate crude bound of
    $\mU_T(\SEGB') \geq \left(  1-\frac{KH}{T}\right) \OPT_T$.
    \item This proof uses the analysis presented in Section~\ref{sec:thm1 outline}.
\end{enumerate}

The proof is composed of two steps. In the first step, we show that if $T > T_0$ for some $T_0$, any optimal algorithm must explore the arms according to a policy that admits the same structure of $\OGP$. In the second step, we show that all such policies have an identical exploration time, and hence all achieve the same utility.\\
\textbf{Step 1:} In the case of realizing a positive reward, any algorithm would stop exploring and exploit that realized reward. Consequently, we can separate exploration rounds from exploitation rounds. Notice, however, that the exploration policy can select portfolios different that $\OGP$ for finite $T$. To illustrate, reconsider Example~\ref{example with four}. In the extreme case of $T=1$, there is no point in playing $\bl p_{1,3}$, since exploring $a_3$ is futile; we only care about maximizing the current round's reward.

However, if $T$ is large \textit{enough}, any optimal algorithm must employ an asymptotically optimal policy. To see this, let $(\pi,\ALG^{\pi})$ be a pair of exploration policy and the algorithm that employs it, and assume $\pi$ does not admit the structure of $\OGP$. The utility of $\ALG^\pi$ satisfies
\begin{align}\label{eq:alg of pi is not optimal}
\mU_T(\ALG^\pi) \leq KH+ (T-K)W(\pi,A).    
\end{align}
Similarly, taking into account the optimality of $\OGP$,
\begin{align}\label{eq:alg of pi with ogp}
\mU_T(\ALG^{\OGP}) \geq -K +(T-K)W^*(A).
\end{align}
Using similar arguments to those in Proposition~\ref{prop:optimal p valid}, we can assume w.l.o.g. that $\pi$ belongs to $\{2^A \rightarrow \mP \cup \mP'\}$ (recall the definition of $\mP$ and $\mP'$ from Subsection~\ref{subsec:bin}). To see this, observe that any safe policy can be formulated as a convex combination of policies that use $\mP \cup \mP'$ solely, and therefore we can assume that $\pi$ is the one for which $\ALG^{\pi}$ gets the highest utility. Furthermore, due to the proof of Theorem~\ref{thm:holy grail} (precisely Equation~\eqref{eq:rho with alpha}) it follows that if $\pi$ does not admit the structure of $\OGP$, then it is strictly sub-optimal. Next, let
\begin{align}\label{eq:alg of pi little omega}
\omega  \defeq \min_{\substack{\rho \in \{2^A \rightarrow \mP \cup \mP'\},\\W(\rho, A) < W^*(A)}} W^*(A) - W(\rho,A) > 0.
\end{align}
The quantity $\omega$ concerns the distributions of $(X(a_i))_i$ and is completely independent of the time $T$. We can further quantify or bound $\omega$ but this abstract and simple form is sufficient for our purposes. Let $T_0 \defeq K + \frac{K(H+1)}{\omega}$. Combining Inequalities~ \eqref{eq:alg of pi is not optimal},\eqref{eq:alg of pi with ogp}, and \eqref{eq:alg of pi little omega} we get
\begin{align*}
\mU_T(\ALG^{\OGP})-\mU_T(\ALG^\pi) &\geq  -K +(T-K)W^*(A) - KH - (T-K)W(\pi,A) \\
& \geq -K(H+1) + (T-K)\omega\\
& > 0,
\end{align*}
provided that $T>T_0$. To conclude this step, we know that $\pi$ is a variation of $\OGP$.

\textbf{Step 2:}
Notice, however, that $\OGP$ is a class of policies differing from one another in the choices of arms from $\above(A)$ (Line~\ref{policy:pick arbitrary} in Policy~\ref{policy:pi star}); hence, one policy may attain a better reward than the other by reaching exploitation faster. 

As we commented in Line~\ref{algpi:play with ogp} of Algorithm~\ref{alg:alg of pi two supported}, when the state $s$ does not contain arms from $\below(A)$ we prioritize singleton portfolios according to the stochastic order. That is, we favor $\bl p_{i,i}$ over $\bl p_{i',i'}$ if $\mu(a_i) > \mu(a_{i'})$. This modification ensures that the time to exploitation from such states is minimal. 

Nevertheless, we might face a problem in states for which $\below(s)\neq\emptyset$. To illustrate, consider the action $\bl p_{i,j}$ for some $a_i\in \above(A), a_j\in \below(A)$. The probability we discover a reward of $H$ when playing $\bl p_{i,j}$ is
\begin{align}\label{eq:q is identical to all}
\bl p_{i,j}(a_i)\Pr(X(a_i)=H)  + \bl p_{i,j}(a_j)\Pr(X(a_j)=H).
\end{align}
Consequently, we might favor $\bl p_{i,j}$ over $\bl p_{i',j}$ (for $a_{i'}\in \above(A)$) if it allows faster discovery of a positive reward (which is necessarily $H$). However, as we show next, the probability in Equation~\eqref{eq:q is identical to all} is the same regardless of the selection the arm from $\above(s)$. Observe that 
\[
\mu(a_i) = H\Pr(X(a_j)=H) +(-1)\cdot (1-\Pr(X(a_j)=H)) = (H+1)\Pr(X(a_j)=H)-1; 
\]
thus, by reformulating Equation~\eqref{eq:q is identical to all} we get
{
\thinmuskip=.2mu
\medmuskip=0mu plus .2mu minus .2mu
\thickmuskip=1mu plus 1mu
\begin{align*}\label{eq:q is identical to all elaborate}
\textnormal{Eq.}\eqref{eq:q is identical to all}&=\frac{-\mu(a_j)}{\mu(a_i)-\mu(a_j)}\Pr(X(a_i)=H)+ \frac{\mu(a_i)}{\mu(a_i)-\mu(a_j)}\Pr(X(a_j)=H)\\
& =\frac{-(H+1)\Pr(X(a_j)=H)+1}{\mu(a_i)-\mu(a_j)}\Pr(X(a_i)=H)+ \frac{(H+1)\Pr(X(a_i)=H)-1}{\mu(a_i)-\mu(a_j)}\Pr(X(a_j)=H)\\
& = \frac{\Pr(X(a_i)=H)-\Pr(X(a_j)=H)}{\mu(a_i)-\mu(a_j)}\\
& = \frac{\Pr(X(a_i)=H)-\Pr(X(a_j)=H)}{(H+1)\Pr(X(a_i)=H)-1-(H+1)\Pr(X(a_i)=H)+1}\\
& = \frac{1}{H+1}.
\end{align*}
}%

\begin{algorithm}[t]
\renewcommand{\algorithmiccomment}[1]{\texttt{\kibitz{blue}{\##1}}}
\caption{$\SEGB$ for Two-Supported Distributions \label{alg:alg of pi two supported}}
\begin{algorithmic}[1]
\STATE $s\gets A$
\WHILE[$s$ is not a terminal state] {$\OGP(s)\neq \emptyset$\label{algpi:while}}{
\STATE play $\OGP(s)$, and denote the realized action by $a_k$.\label{algpi:play with ogp} \COMMENT{if $\below(s)=\emptyset$, prioritize the arms in $\above(s)$ according to the stochastic order}
\IF[a reward of $H$ was realized]{$x_{a_k}>0$} {
		\STATE \textbf{break}. 
}
\ENDIF
\STATE $s\gets s\setminus \{a_k\}$.\label{algpi:update s}
}
\ENDWHILE
\IF{$x(a_k)=H$ for some explored arm $a_k$}{
\STATE exploit $a_k$ forever.
}
\ELSE{
\STATE exploit $a_0$ forever.
}
\ENDIF
\end{algorithmic}
\end{algorithm}
Since all $\OGP$ have the same expected exploration time, they all achieve the same utility. This completes the proof of Proposition~\ref{prop:bernoulli opt}.
\end{proofof}

\begin{proofof}{Proposition \ref{prop:i-d bounds}}
To prove the claim, we show that $\SEGB$ explores for at most $K(1+\frac{\gamma}{\delta})$ rounds, and then exploits. First, $\SEGB$ uses $K_1 \leq K$ rounds following $\OGP$ until it reaches a terminal state (Lines~\ref{algpi:while}--\ref{algpi:update s}). If all the realized rewards are negative, the exploration ends after then. Otherwise, if it discovered a positive reward, the value of that reward is at least $\delta$. Next, the Bernoulli trails will explore every unexplored arm w.p. of at least $ \frac{\delta}{\delta+\gamma}$, or $ \frac{\delta+\gamma}{\delta}$ rounds in expectation. Therefore, after $(K-K_1)(1+\frac{\gamma}{\delta})$ rounds in expectation we explore all the remaining arms. Overall, the expected number of rounds devoted to exploration is
\[
K_1 + (K-K_1)\left(1+\frac{\gamma}{\delta}\right) \leq K \left(1+\frac{\gamma}{\delta}\right).
\]
Ultimately, recall that $\lim_{T \rightarrow \infty }\mU_T(\SEGB) = \lim_{T \rightarrow \infty }\OPT_T$, so $\SEGB$ exploits an expected reward of $\lim_{T \rightarrow \infty }\OPT_T$ after it completes its exploration. Therefore,
\begin{align*}
\mU_T(\SEGB)& \geq K \left(1+\frac{\gamma}{\delta}\right) \cdot 0 + \left(T-K \left(1+\frac{\gamma}{\delta}\right) \right)\lim_{T \rightarrow \infty }\OPT_T. \\
&\geq\left(T-K \left(1+\frac{\gamma}{\delta}\right) \right)\OPT_T. 
\end{align*}
\end{proofof}

%% file: input/proposition_p_valid.tex
\begin{proofof}{Proposition \ref{prop:optimal p valid}}
Fix a non-terminal state $U\subseteq A$. Further, denote $V(a)\defeq W^*(U\setminus \{a\})$ for every $a\in U$. Due to Equation (\ref{eq:W elaborated}), the action that maximizes the reward at state $U$ is the solution $\bl p \in \Delta(U)$ of the following linear program:
\begin{equation}
\tag{P1} \label{eq:lp w}
\begin{array}{ll@{}ll}
\max \limits_{\bl p} &\sum_{a\in U} \bl p(a)V(a)   & \\
\text{subject to} & \sum_{a\in U} \bl p(a)\mu(a) \geq 0 &     \\
& \sum_{a\in U} \bl p(a)=1 &  \\
& 0\leq \bl p(a) \leq 1 & \textnormal{for all }a\in U
\end{array}
\end{equation}
Observe that for every $\bl p$ such that $ \sum_{a\in U} \bl p(a)\mu(a) \geq 0$, there exist coefficients $(\alpha_{i,j})_{i,j}$ such that for every $a_i\in U$, $\bl p(a_i)=\sum_{a_j\in U}\alpha_{i,j}\bl p_{i,j}(i)$, and
\[
\sum_{a\in U} \bl p(a)\mu(a) = \sum_{i,j}\alpha_{i,j}\left(\bl p_{i,j}(i)\mu(a_i)+\bl p_{i,j}(j)\mu(a_j)\right).
\]
Hence, an equivalent form of Problem (\ref{eq:lp w}) is
\begin{equation}
\tag{P2} \label{eq:lp w with alpha}
\begin{array}{ll@{}ll}
\max \limits_{\bl \alpha} &\sum_{a\in U} \alpha_{i,j}\left(\bl p_{i,j}(i)V(a_i)+\bl p_{i,j}(j)V(a_j)\right)  & \\
\text{subject to} & \sum_{i,j}\alpha_{i,j}\left(\bl p_{i,j}(i)\mu(a_i)+\bl p_{i,j}(j)\mu(a_j)\right) \geq 0 &     \\
&\sum_{i,j}\alpha_{i,j}=1 &  \\
& 0\leq \alpha_{i,j} \leq 1 \qquad  \textnormal{for all }(i,j)\in \{(i',j')\mid \bl p_{i',j'} \in \mP\cup \mP'\textnormal{ and } i',j' \in U \} &
\end{array}
\end{equation}
Finally, notice that the constraint $\sum_{i,j}\alpha_{i,j}\left(\bl p_{i,j}(i)\mu(a_i)+\bl p_{i,j}(j)\mu(a_j)\right) \geq 0$ holds for every selection of $(\alpha_{i,j})_{i,j}$ by the way we defined $\mP\cup \mP'$; thus, the maximum of Problem (\ref{eq:lp w with alpha}) is obtain when we set $\alpha_{i,j}=1$ for the pair $(i,j)$ that maximizes $\left(\bl p_{i,j}(i)V(a_i)+\bl p_{i,j}(j)V(a_j)\right) $.
\end{proofof}

%% file: input/auxiliary_statements.tex
\begin{claim}\label{claim:ass is not for W}
Consider a state $U\in\mS$, such that $\below(U) \geq 2$. Let  $a_j = \argmin_{a_{j'}\in\below(U)}\sigr_\pi(a_{j'})$, and let $a_{\tilde j}\in \below(U), a_{\tilde j} \neq a_j$. Under Assumption \ref{assumption:dominance}, it might be the case that $W^*(U\setminus \{a_j\}) < W^*(U\setminus \{a_{\tilde j}\})$.
\end{claim}
\begin{proofof}{Claim \ref{claim:ass is not for W}}
We prove the claim by providing an example, that could be easily extended to a family of infinitely many examples. Consider $K=3$, $A=\{a_1,a_2,a_3\}$ such that 
\[
X_1=\begin{cases}
-1 & \textnormal{w.p. 0.45}\\
1 & \textnormal{w.p. 0.55}
\end{cases}, \qquad
X_2=\begin{cases}
-10^6-2\epsilon & \textnormal{w.p. 0.5}\\
10^6 & \textnormal{w.p. 0.5}
\end{cases}, \qquad
X_3=\begin{cases}
-10^{\frac{1}{\epsilon}} & \textnormal{w.p. 0.5}\\
10^6 & \textnormal{w.p. 0.5}
\end{cases}
\]
For small $\epsilon$, say $\epsilon<\frac{1}{7}$, it is clear that $X_2$ stochastically dominates $X_3$. The resulting expected values are $\mu({a_1})=0.1,\mu({a_2})=-\epsilon,$ and $\mu({a_3})= -\Theta(10^{\frac{1}{\epsilon}})$. The intuition behind our selection of rewards is that arm $a_2$ could have high reward, and can be explored with probability $\bl p_{1,2}(2)=1-O(\epsilon)$. On the other hand, arm $a_3$ has a high reward with the same probability, but it is highly unlikely to explore it. More precisely, $\bl p_{1,3}(3)=\Theta(10^{-\frac{1}{\epsilon}})$. To finalize the argument, notice that
\[
W^*(A\setminus \{a_2\})=\bl p_{1,3}(1)\cdot R(\{a_3\})+\bl p_{1,3}(3)\cdot \bl p_{1,1}(1)\cdot R(\emptyset)=0.5\cdot 10^6+0.5\cdot 0.55\cdot 1 +O(\epsilon)
\]
while
\begin{align*}
W^*(A\setminus \{a_3\})&=\bl p_{1,2}(1)\cdot R(\{a_2\})+\bl p_{1,2}(2)\cdot \bl p_{1,1}(1)\cdot R(\emptyset)\\
&=0.75\cdot 10^6+0.25\cdot 0.55\cdot 1+O(\epsilon).
\end{align*}
The proof is completed by taking $\epsilon$ to zero.
\end{proofof}